\documentclass[12pt]{article}
\newcommand{\blind}{1}
\RequirePackage[OT1]{fontenc}
\usepackage[dvipsnames]{xcolor}
\RequirePackage{amsthm,amsmath,amsbsy,amssymb,array,bm,graphicx,epstopdf,mathtools,calc,authblk,anyfontsize,booktabs,multirow,psfrag,epsf,enumitem}
\usepackage{tikz,pgf}
\usetikzlibrary{decorations,decorations.markings}
\usepackage{ctable} 
\usepackage[activate={true,nocompatibility},final,tracking=true,kerning=true,spacing=true,factor=1100,stretch=10,shrink=10,spacing=basictext]{microtype}
\usepackage{natbib}
\usepackage{xr-hyper}
\RequirePackage[colorlinks=true,citecolor=MidnightBlue,urlcolor=MidnightBlue,linkcolor=PineGreen,breaklinks]{hyperref}
\usepackage[noabbrev,capitalise,nameinlink]{cleveref}
\allowdisplaybreaks

\tikzstyle{dir}= [postaction={decorate,
	decoration={markings,
	mark=at position .65
	with {\arrow[scale=1]{stealth}}}}]
\tikzstyle{dirs}= [postaction={decorate,
	decoration={markings,
	mark=at position .65
	with {\arrow[scale=.8]{stealth}}}}]
\tikzstyle{nd} = [circle, fill=black,
	inner sep=0pt,
	minimum width=3 pt]
\tikzstyle{bnd} = [circle,fill=black,
	inner sep=1pt,
	minimum width=6 pt,
	text=white]
\tikzstyle{rnd} = [circle, fill=black!30,
	inner sep=0pt,
	minimum width=3 pt]
\tikzstyle{brnd} = [circle,fill=black!30,
	inner sep=1pt,
	minimum width=6 pt,
	text=black]
\newcommand{\Textedge}{%
\raisebox{0.3mm}{\!\!
\tikz[scale=.2,clip]{\draw[thick]
(180:.6) node[nd] {}--
(  0:.6) node[nd]{};
}}\hspace{-1.4mm}
}
\newcommand{\Texttriangle}{%
\raisebox{-0.45mm}{\!\!
\tikz[scale=.2,clip]{\draw[thick]
(210:.6) node[nd] {}--
(  90:.6) node[nd]{}--
( -30:.6) node[nd]{}--
(210:.6) node[nd]{};
}}\hspace{-.9mm}
}

\newcommand{\Textsquare}{%
\raisebox{-0.45mm}{\!
\tikz[scale=.2,clip]{\draw[thick]
(        45:.65) node[nd]{}--
(  45+90:.65) node[nd]{}--
(45+180:.65) node[nd]{}--
(45+270:.65) node[nd]{}--
(        45:.65) node[nd]{};
}}\hspace{-.9mm}
}
\newcommand{\Textdiamond}{%
\raisebox{-0.45mm}{\!
\tikz[scale=.2,clip]{\draw[thick]
(        45:.65) node[nd]{}--
(  45+90:.65) node[nd]{}--
(45+180:.65) node[nd]{}--
(45+270:.65) node[nd]{}--
(        45:.65) node[nd]{}--
(45+180:.65) node[nd]{};
}}\hspace{-.5mm}
}
\newcommand{\Textbowtie}{%
\raisebox{-0.45mm}{\!
\tikz[scale=.2,clip]{\draw[thick]
(        30:.75) node[nd]{}--
(          0:   0) node[nd]{}--
(30+180:.75) node[nd]{}--
(-30-180:.75) node[nd]{}--
(          0:   0) node[nd]{}--
(       -30:.75) node[nd]{}--
(        30:.75) node[nd] {};
}}
}
\newcommand{\Texthourglass}{%
\raisebox{-0.45mm}{\!
\tikz[scale=.2,clip]{\draw[thick]
(        45:.75) node[nd]{}--
(          0:   0) node[nd]{}--
(45+180:.75) node[nd]{}--
(      180:1.1) node[nd]{}--
(-45-180:.75) node[nd]{}--
(          0:   0) node[nd]{}--
(       -45:.75) node[nd]{}--
(          0:1.1) node[nd]{}--
(        45:.75) node[nd] {};
}}
}
\newcommand{\Textdblsqr}{%
\raisebox{-0.7mm}{\!
\tikz[scale=.2,clip]{\draw[thick]
(        90:.70) node[nd]{}--
(          0:.45) node[nd]{}--
(       -90:.70) node[nd]{}--
(      180:.45) node[nd]{}--
(        90:.70) node[nd]{}--
(        90:.70) node[nd]{}--
(          0:1.1) node[nd]{}--
(       -90:.70) node[nd]{}--
(      180:1.1) node[nd]{}--
(        90:.70) node[nd]{};
}}
}
\newcommand{\TextdblsqrBis}{%
\raisebox{-0.7mm}{\!
\tikz[scale=.2,clip]{\draw[thick]
(        45:.75) node[nd]{}--
(       -45:.75) node[nd]{}--
(     -135:.75) node[nd]{}--
(      135:.75) node[nd]{}--
(        45:.75) node[nd]{}--
(          0:   0) node[nd]{}--
(       -45:.75) node[nd]{}--
(          0:1.1) node[nd]{}--
(        45:.75) node[nd]{};
}}
}
\newcommand{\TextdblsqrTer}{%
\raisebox{-0.7mm}{\!
\tikz[scale=.17,clip]{\draw[thick]
(0,0) node[nd]{}--
(1,0) node[nd]{}--
(1,1) node[nd]{}--
(2,1) node[nd]{}--
(2,0) node[nd]{}--
(1,0) node[nd]{}--
(1,1) node[nd]{}--
(0,1) node[nd]{}--
(0,0) node[nd]{};
}}
}
\newcommand{\TextdblsqrQuad}{%
\raisebox{-0.7mm}{\!
\tikz[scale=.2,clip]{\draw[thick]
(0,0) node[nd]{}--
(1,0) node[nd]{}--
(2,0) node[nd]{}--
(1.5,.87) node[nd]{}--
(1,0) node[nd]{}--
(.5,.87) node[nd]{}--
(0,0) node[nd]{}--
(.5,.87) node[nd]{}--
(1.5,.87) node[nd]{};
}}
}
\newcommand{\TextdblsqrQint}{%
\raisebox{-0.7mm}{\!
\tikz[scale=.18,clip]{\draw[thick]
(        90:.75) node[nd]{}--
(          0:.75) node[nd]{}--
(       -90:.75) node[nd]{}--
(      180:.75) node[nd]{}--
(        90:.75) node[nd]{}--
(          0:   0) node[nd]{}--
(       -90:.75) node[nd]{};
}}
}
\newcommand{\TextdblsqrSext}{%
\raisebox{-0.7mm}{\!
\tikz[scale=.18,clip]{\draw[thick]
(0,0) node[nd]{}--
(1,0) node[nd]{}--
(1,1) node[nd]{}--
(0,1) node[nd]{}--
(0,0) node[nd]{}--
(1,1) node[nd]{}--
(0,1) node[nd]{}--
(1,0) node[nd]{};
}}
}

\theoremstyle{plain}
\crefname{equation}{}{}
\crefname{figure}{Figure}{Figures}
\crefname{table}{Table}{Tables}
\newtheorem{Definition}{Definition}
\crefname{Definition}{Definition}{Definition}
\newtheorem{Lemma}{Lemma}
\crefname{Lemma}{Lemma}{Lemmas}
\newtheorem{Proposition}{Proposition}
\crefname{Proposition}{Proposition}{Propositions}
\newtheorem{Theorem}{Theorem}
\crefname{Theorem}{Theorem}{Theorems}
\newtheorem{Corollary}{Corollary}
\crefname{Corollary}{Corollary}{Corollarys}
\newtheorem{Remark}{Remark}
\crefname{Remark}{Remark}{Remarks}
\newcommand{\cov}{\ensuremath{\mathrm{cov}}}
\newcommand{\E}{\ensuremath{\mathbb{E}}}
\newcommand{\var}{\ensuremath{\mathrm{var}}}
\newcommand{\D}{\ensuremath{\mathcal{D}}}
\newcommand{\F}{\ensuremath{\mathcal{F}}}

\newcommand{\G}{\ensuremath{\mathcal{G}}}
\newcommand{\X}{\ensuremath{\mathcal{X}}}
\renewcommand{\P}{\ensuremath{\operatorname{\mathbb{P}}}}
\allowdisplaybreaks
\makeatletter
\newcommand{\pushright}[1]{\ifmeasuring@#1\else\omit\hfill$\displaystyle#1$\fi\ignorespaces}
\newcommand{\pushleft}[1]{\ifmeasuring@#1\else\omit$\displaystyle#1$\hfill\fi\ignorespaces}
\makeatother
\renewcommand\footnotemark{}
\date{}
\begin{document}
\def\spacingset#1{\renewcommand{\baselinestretch}%
{#1}\small\normalsize} \spacingset{1}
\newcommand{\TheTitle}{Testing for Equivalence of Network\\[.1\baselineskip] Distribution Using Subgraph Counts}
\if1\blind
{
 \title{\bf \TheTitle}
 \author{P-A. G. Maugis\thanks{
 This work was supported in part by DARPA under SIMPLEX contract N66001-15-C-4041; by the US Army Research Office under Multidisciplinary University Research Initiative Award 58153-MA-MUR; by the US Office of Naval Research under Award N00014-14-1-0819; by the UK Engineering and Physical Sciences Research Council under Mathematical Sciences Leadership Fellowship EP/I005250/1, Established Career Fellowship EP/K005413/1, Developing Leaders Award EP/L001519/1, and Award EP/N007336/1; by the UK Royal Society under a Wolfson Research Merit Award; and by Marie Curie FP7 Integration Grant PCIG12-GA-2012-334622 and the European Research Council under Grant CoG 2015-682172NETS, both within the Seventh European Union Framework Program. The authors thank the Isaac Newton Institute for Mathematical Sciences, Cambridge, UK, for support and hospitality during the program Theoretical Foundations for Statistical Network Analysis (EPSRC grant no. EP/K032208/1) where a portion of the work on this paper was undertaken. The authors also thank Joshua T. Vogelstein and his team for connectome data and expertise.}\hspace{.2cm}\\
  Department of Statistical Science, University College London\\
  and\\
  S. C. Olhede\\
  School of Basic Sciences, \'Ecole Polytechnique F\'ed\'erale de Lausanne\\
  and\\
  C. E. Priebe\\
  Department of Applied Mathematics and Statistics, Johns Hopkins University
  and\\
  P. J. Wolfe\\
  Department of Statistics, Purdue University}
 \maketitle
 \newpage
} \fi

\if0\blind
{
 \bigskip
 \bigskip
 \bigskip
 \begin{center}
  {\LARGE\bf \TheTitle}
\end{center}
 \medskip
 \bigskip
 \newpage
} \fi
\begin{abstract}
We consider that a network is an observation, and a collection of observed networks forms a sample. In this setting, we provide methods to test whether all observations in a network sample are drawn from a specified model. We achieve this by deriving the joint asymptotic properties of average subgraph counts as the number of observed networks increases but the number of nodes in each network remains finite. In doing so, we do not require that each observed network contains the same number of nodes, or is drawn from the same distribution. Our results yield joint confidence regions for subgraph counts, and therefore methods for testing whether the observations in a network sample are drawn from: a specified distribution, a specified model, or from the same model as another network sample. We present simulation experiments and an illustrative example on a sample of brain networks where we find that highly creative individuals' brains present significantly more short cycles than found in less creative people.
\end{abstract}

\noindent%
\vfill

\newpage
\spacingset{1.45} 
\section{Introduction}

We show that subgraph counts are flexible and powerful statistics for testing distributional properties of networks, when more than one network is observed. Specifically, we use subgraph counts to test the hypotheses that all networks in a sample are generated either from a given distribution, from distributions specified by a given model class, or from the same model as that of another sample.

Our results address the fundamental inference problem raised by the following experiment~\citep{gray2012magnetic}: The networks connecting brain regions of individuals of varied levels of creativity are observed. However, while these observations can be assumed to be independent, due to the variability of the brain structure and the instability of the observation technique, they cannot be assumed to be identically distributed; for instance, they need not contain as many nodes and edges. Unfortunately, this implies that if we identify each network realization with its adjacency matrix, the obtained matrices will be of different sizes. This prevents, for instance, estimating the distribution the adjacency matrices are drawn from as one would from a sample of independent and identically distributed realizations. How, while allowing for such variations, can we test for significant differences between individuals with different levels of creativity?

Formally, we consider that each network is an observation---say $G_i$ for each subject---and a collection of observed networks form a sample---say $\G = (G_1,\dots,G_N)$. Then, our goal is to compare distributional properties of the $G_i$-s as $N$ grows. This parallels more classical statistical settings, where an observation is a vector---such as ${\bm X}_i\in\mathbb{R}^k$---and a sample is a matrix: $\X = ({\bm X}_1,\dots,{\bm X}_N)\in\mathbb{R}^{k\times N}$. However, this parallel with the classical setting stops there. Indeed, the $G_i$ need not be of the same size, or have nodes that are matched; i.e., the $X_i$-s would not have the same dimension, and the entries could be shuffled.

This setting strongly differs from the two settings that have already seen extensive research. First, the setting where only one or two very large network are observed and for which many statistical tests already exist (see~\cite{hoff2012latent,ho2012multiscale,BickelLevina2012,Sussman2012,bhattacharyya2013subsampling,tang2013,olhede2013network,fosdick2015testing,klopp2016, coulson2016poisson}, to cite but a few). Available statistical tests for network comparison in that setting focus on the case where one or two large networks are observed~\citep{asta2014geometric,tang2016nonpara,gao2017testing,banerjee2017testing,
ghoshdastidar2017testing}, or when one finite network is compared to a fixed model alone~\citep{birmele2012detecting,Ali2014alignment,ali2016comparison}. The second case addresses graph samples, as we do here, but under very restrictive assumptions: samples that are independent and identically distributed, where all graphs have the same size, and where nodes across networks can be matched one-to-one. Then, under these assumptions, it is possible to compare network summaries using classical statistical methods~\citep{simpson2012exponential,stoffers2013,daianu2013,ginestet2017}, or by fitting a parametric or semi-parametric model~\citep{simpson2012exponential,wang2018,durante2018}. 

We provide, in the graph sample setting described above, an analog of the multivariate $t$-test for network samples: methods to test whether a given network sample $\G$ presents {\em averages} consistent with either a specific model, or with that of another sample. The averages we use are {\em subgraph counts}; e.g., the number of \Texttriangle~or \Textsquare~in the sample. The choice of subgraph counts as statistics is motivated by their success in comparing large networks~\citep{milo2002network,Ali2014alignment}, but also by results in random graph theory and the study of large graphs. In both fields, subgraph counts have proved to be the most powerful tool available to compare networks~\citep{Ali2014alignment,ali2016comparison}
and are known to have properties similar to moments of random variables when studying large networks~\citep{diaconis2008graph,BickelLevina2012,lovasz2012large,chatterjee2013estimating}. Finally, on account of these results, many powerful algorithm exists to count subgraphs efficiently; e.g., ~\cite{talukder2016distributed,ortmann2016quad}.

Formally, we are representing network samples in a space defined by subgraph counts, and performing comparisons in that space. While other network comparison techniques also use embeddings~\citep{asta2014geometric,tang2016nonpara,gao2017testing,wang2018}, using subgraph counts presents three key advantages: first, if the $G_i$-s are generated by a blockmodel~\citep{hoff2012latent}---one of the most popular random network models to date---and for some families of subgraphs, the embedding is one-to-one. This result is known as the finite forciblity of a family of graphs~\citep[Ch.~16]{lovasz2012large}. Second, very few assumptions on each $G_i$ need to be made as $N$ grows to obtain consistency and asymptotic normality of the image of $\G$ in the embedding space. This enables us to work under a very flexible null model. Finally, because it relies on direct summaries of the $G_i$-s---the number of~\Texttriangle,~\Textsquare, and so on---this embedding remains interpretable, part of the appeal to use subgraph based inference.

The main risk in subgraphs for testing purposes is that we cannot be assured that the subgraphs considered are sufficient to distinguish the null and the data generative mechanism. We provide several experiment allowing us to claim that while possible, this does not appear to be common. Furthermore, this risk is balanced by the interpretability of subgraph counts: if the null and the generative mechanism have the same number of subgraph, they are maybe equivalent for the purpose of the study at hand. Another shortcoming is that subgraph counts can be highly correlated, especially in denser networks, making the estimation of the inverse of their covariance matrix unstable. One of our contribution are closed form formulae for these covariance matrices under our null, which reduces but does not fully resolves this issue. Finally, because subgraph counting library are not standard, implementing the proposed methods is not as easy as for other methods, which tend to rely on more common, linear algebra related, data transformation pipelines. To this end we made all the code to perform our analysis available\footnote{Code to perform all methods, as well as reproduce all figure and tables reported in the article, can be found at: \url{https://github.com/PierreAndreMaugis/motifs-for-network-samples}}.

In the remainder of this article, we first introduce subgraph counts and the kernel based random graph model. We then present, successively, the case where all the networks in the sample come from the same kernel model, and the case where each observed network may come from a different kernels. In both cases, we prove asymptotic normality of our estimator, present a plug-in estimator of its variance, describe the limit of the estimator under the alternative hypothesis, and provide representative examples showing the practical utility of the result. Methods to produce the figures, as well as supporting simulations, can be found in the Supplementary Material. We conclude with an analysis of connectome data, and with a discussion.

\section{Subgraph Counts in Kernel Based Random Graphs}

We now define our statistics (subgraph counts) and our null model (the kernel based random graph model). Subgraph counts are natural statistics to compare networks for two reasons. First, subgraph counts intuitively summarize a network through its fundamental building blocks. This has historically given them purchase to address hard fundamental and empirical problems~\citep{rucinski1988small,milo2002network,Ali2014alignment}. Second, subgraph counts present tractable analytical properties. We will describe and leverage these properties below, in a manner paralleling what is done in related literature~\citep{BickelLevina2012,bhattacharyya2013subsampling,rucinski1988small}.

\begin{figure}
\centering
\def\s{1.35}
	\def\x{6}
	\tikzset{ultra thick/.style={line width=\s pt}}
	\tikzstyle{n0}=[circle, draw=black, fill=black,
			inner sep=0pt, minimum width=\s*\x pt]
	\tikzset{e0/.style={ultra thick,color=black}}

(a)$\ $
\begin{tikzpicture}[scale=\s]
	\draw  {(0,0) node[n0] {} edge[e0] (1,1)};
	\draw  {(0,0) node[n0] {} edge[e0] (0,1)};
	\draw  {(0,0) node[n0] {} edge[e0] (1,0)};
	\draw  {(0,1) node[n0] {} edge[e0] (1,1)};
	\draw  {(1,0) node[n0] {} edge[e0] (1,1)};
	\draw  {(1,1) node[n0] {} edge[e0] (2,1)};
	\draw  {(2,1) node[n0] {} edge[e0] (1,1)};
\end{tikzpicture}
$\qquad$(b)$\ $
\begin{tikzpicture}[scale=\s]
	\draw  {(0,0) node[n0] {} edge[e0] (1,1)};
	\draw  {(0,0) node[n0] {} edge[e0] (0,1)};
	\draw  {(0,1) node[n0] {} edge[e0] (1,1)};
	\draw  {(1,0) node[n0] {} edge[e0] (1,1)};
	\draw  {(1,1) node[n0] {} edge[e0] (2,1)};
	\draw  {(2,1) node[n0] {} edge[e0] (1,0)};
\end{tikzpicture}
$\qquad$(c)$\ $
\begin{tikzpicture}[scale=\s]
	\draw  {(0,0) node[n0] {} edge[e0] (.75,.5)};
	\draw  {(0,0) node[n0] {} edge[e0] (0,1)};
	\draw  {(.75,.5) node[n0] {} edge[e0] (1.5,1)};
	\draw  {(1.5,1) node[n0] {} edge[e0] (1.5,0)};
	\draw  {(0,1) node[n0] {} edge[e0] (1.5,1)};
	\draw  {(1.5,0) node[n0] {} edge[e0] (0,0)};
\end{tikzpicture}
\vspace{-.5\baselineskip}
\caption{\label{GraphEx}Example of subgraph counts. There are $6$ copies of the edge (\protect\Textedge) in all three graphs. There are $2$ copies of the triangle (\protect\Texttriangle) in a) and b), but $0$ in c). There are $1$, $0$ and $3$ copies of the square (\protect\Textsquare) in a), b) and c) respectively.}
\end{figure}

A \emph{subgraph count} is the number of copies of a given graph in another graph (see \cref{GraphEx}). Throughout, we call the subgraph---and denote $F$---the graph which is counted and $G$ the larger graph in which the counting takes place. All graphs will be simple (unweighted, no self loops or multiple edges). Subgraphs are also termed motifs, pattern graphs or shapes depending on the field~\citep{milo2002network,alon1997cycles,hocevar2014,benson2016higher}.

For clarity, we define subgraph counts formally as follows:
\begin{Definition}[Graph equivalence `$\equiv$']
Fix two graphs $F$ and $F'$. We say that $F$ is equivalent to---or is a copy of---$F'$, and write $F\equiv F'$, if there exists a bijective map $\phi$ from the vertex set of $F$ to the vertex set of $F'$ such that $ij$ is an edge in $F$ iff $\phi(i)\phi(j)$ is an edge in $F'$.
\end{Definition}
\begin{Definition}[Subgraph count $X_F(G)$]\label{SubgraphDef}
Fix two graphs $F$ and $G$. We denote $X_F(G)$ the number of non-necessarily induced subgraphs of $G$ equivalent to $F$; i.e,
\[
X_F(G) = \#\left\{F'\subset G : F'\equiv F\right\},
\]
where $F\subset G$ if the vertex end edge sets of $F$ are subsets of those of $G$.
\end{Definition}
With this notation, calling $G_a$, $G_b$ and $G_c$ the graphs in \cref{GraphEx}, we have $X_{\!\Textsquare}\ (G_a)=1$, $X_{\!\Textsquare}\ (G_b)=0$, $X_{\!\Textsquare}\ (G_c)=3$. A more complete definition is provided in \cref{SubgraphDef}.

The power of subgraph counts to study networks stems from their inherent linearity. Indeed, products of subgraph counts are but linear combinations of other subgraph counts. Intuitively, first observe that a product of two subgraph counts will involve counting pairs of copies. Then, a product of subgraph counts can be recovered by counting the number of copies of all subgraphs that can be induced by a pair of copies. More precisely, in the Appendix we show the following, which is implicitly used in \citet{rucinski1988small}:
\begin{Lemma}[Linearity of subgraph counts \citep{rucinski1988small}]\label{LinearityLemma}
For any two graphs $F$ and $F'$, there are factors $c_H$ and a set $\mathcal{H}_{FF'}$ of subgraphs---the set of subgraphs that can be obtained using one copy of each $F$ and $F'$ as building blocks---such that for any graph $G$
\[
X_{F_1}(G)X_{F_2}(G) = \sum_{H\in\mathcal{H}_{FF'}} c_HX_H(G).
\]
\end{Lemma}
For instance, as these will be used later on, we have that in {\emph{any}} graph $G$,
\begin{align*}
X_{\!\Texttriangle\,}(G)^2 &=
	2X_{\!\Texttriangle\ \Texttriangle\,}(G)
	+2X_{\!\Textbowtie}(G)
	+2X_{\!\Textdiamond\,}(G)
	+ X_{\!\Texttriangle\,}(G),\\
X_{\!\Textsquare\ }(G)^2 &=
	2X_{\!\Textsquare\,\Textsquare\ }(G)
	+2X_{\!\Texthourglass}(G)
	+6X_{\!\Textdblsqr}(G)
	+2X_{\!\TextdblsqrBis}(G)\\
&	\qquad+2X_{\!\TextdblsqrTer}(G)
		+2X_{\!\!\TextdblsqrQuad}(G)
		+6X_{\!\!\TextdblsqrQint}(G)
		+6X_{\!\!\TextdblsqrSext}(G)
		+ X_{\!\Textsquare\ }(G).
\end{align*}
This algebraic property of subgraph counts allows us to understand the proofs of~\cite{rucinski1988small,bhattacharyya2013subsampling}. The property is also crucial to the subgraph counting algorithms of~\cite{hocevar2014}, and can be found in many other examples. Crucially, as opposed to cases where the model enforces linearity---such as with assumptions of Normality---it is the nature of the statistics (subgraph counts) and the system (graphs) that makes the problem linear.

The linearity of subgraph counts allows us to use as null the very flexible kernel based model~\citep{lovasz2012large}. This framework subsumes most models used in the statistical literature on networks; e.g., blockmodel~\citep{hoff2012latent} and dot-product models~\citep{Sussman2012}. It has the intuitive structure of affixing to each node $i$ a latent feature (here $x_i$) and of connecting nodes $i$ and $j$ (conditionally independently) with a probability determined by the node features (here $f(x_i,x_j)$).
\begin{Definition}[Kernel $f$ and random graph $G_n(f)$]
Fix a symmetric measurable map $f:[0,1]^2\to[0,1]$, and call it a {\emph kernel}. We call $G_n(f)$ the random graph distribution over graphs with $n$ nodes such that: to each node is randomly and independently assigned a feature $x_i\in[0,1]$, with $x_i\sim\mathrm{Unif}([0,1])$; and where edges form independently conditionally on $\{x_i\}_{i\in[n]}$ with conditional probability
\[
\P[ij\in G| x_i, x_j] = f(x_i,x_j).
\]
\end{Definition}
To recover a blockmodel with $K$ blocks, it suffices to consider a partition of $[0,1]$ in $K$ sets (i.e., $(P_1,\dots,P_K)\in\mathcal{P}_K([0,1])$) and set $f$ as constant over each $P_u\times P_v$. The dot-product model is recovered with a kernel $f$ of finite rank; i.e., $f(x,y)=\sum_{u\leq K}\lambda_u f_u(x)f_u(y)$.

The model assumes that vertices' location in the latent space are independent and identically distributed, so that the model does not assume any structure or symmetry among the nodes in the observed networks. However, it can accommodate any such structure, which would be recovered through estimation. Nonetheless, in some cases it may seem relevant to enforce such structure; e.g., in our brain example, assume that nodes in different hemispheres are unlikely to connect. Doing so risks misspecifying the model to gain faster convergence, which might be necessary in some cases. We did not find it necessary in our application. 

In the kernel framework, subgraph counts have direct interpretation as moments of $f$~\citep{lovasz2012large} (see~\cref{kern_mom}). Specifically, if $G\sim G_n(f)$, then the moments of $X_F(G)$ are moments of $f$. An infinite number of subgraph counts are sufficient statistics to distinguish between any two kernel~\citep{lovasz2012large,bollobas2009metric}, as the subgraph counts can be used to define the subgraph metric. 
However, there are no guarantees on which or how many subgraph counts are needed to distinguish between two kernels. For blockmodels and finite rank models, we know only that a finite number is sufficient (a concept know as finite forcibility, see~\citet[Chapter~16.7 \& Appendix~4]{lovasz2012large} for more details).

This makes subgraph counts especially appropriate as statistics in an hypothesis test. Indeed, for any finite set of subgraphs $\F$, we could have two kernels $f$ and $f'$ such that $f\neq f'$, and yet $\forall F\in\F, \E_{G\sim G_n(f)} X_F(G)=\E_{G\sim G_n(f')} X_F(G)$. However, as we prove below, if $\E_{G\sim G_n(f)} X_F(G)\neq\E_{G\sim G_n(f')} X_F(G)$ for any $F$, then $f\neq f'$. Therefore, a difference in subgraph counts is sufficient but not necessary to distinguish between kernels. Conversely, not observing a difference in subgraph count can only imply that we do not observe a difference in the kernels, and therefore that we fail to reject the hypotheses that the kernel are equal. This has implication regarding the power of our test that we explore below.

Unfortunately, all known results on subgraph counts under the kernel model consider the setting where one very large graph is observed. Here we present the tools to address the problem where a sample of graphs is observed.

\section{The simple case: Samples from one kernel}
We now present a central limit theorem as well as practical methods to build confidence regions for the subgraph counts observed in a network sample $\G=(G_1,\dots,G_N)$. In this section, we assume that there is a kernel $f$ such that each $G_i$ is drawn independently from $G_{n_i}(f)$ (with $n_i=|G_i|$, where for any graph $F$ we write $|F|$ for the number of nodes in $F$).

Fix $F\in\F$ and $G\in\G$. In this setting, $X_F(G)$ is a random variable, and the first parameter to consider is its mean. To compute this mean, let $F_1,\dots F_{X_F(K_G)}$ be all the copies of $F$ in $K_G$ (the complete graph over the nodes of $G$), so that using the linearity of the expectation, we have that
\[
\E X_F(G) = \E \sum_{j\in\left[X_F(K_G)\right]}1_{\{F_j\subset G\}}
= \sum_{j\in\left[X_F(K_G)\right]}\E1_{\{F_j\subset G\}}.
\]
Then, direct computations show that $\E1_{\{F_j\subset G\}}$ does not depend on $j$ (see \cref{first-moms}), and that
\begin{equation}\label{kern_mom}
\E1_{\{F_j\subset G\}} = \mu_F(f) := \int_{[0,1]^{|F|}}\prod_{uv\in F}f(x_u,x_v)\prod_{u\in F}dx_u,
\end{equation}
so that $\E X_F(G) = X_F(K_G)\mu_F(f)$. Observe that $\mu_F(f)$ is a moment of the kernel $f$, as discussed above.

Similar computations for higher moments, aided by \cref{LinearityLemma}, enable us to use the Lindeberg-Feller central limit theorem and the Cramer-Wold device to prove the following:
\begin{Theorem}[Statistical properties of subgraph counts]\label{lyapu}
Fix a set of graphs $\F$, a kernel $f$ and a sequence $n=(n_i)_{i\in\mathbb{N}}$ such that $2\max_{F\in\F}|F|\leq\min_{i\in\mathbb{N}}n_i$. Let $\G=(G_i)_{i\in[N]}$ be a network sample such that for all $i$, $G_i\sim G_{n_i}(f)$. Set $\hat{\mu}_F(\G)=N^{-1}\sum_{G\in\G}{X_F(G)}/{X_F(K_G)}$, $\hat{\bm\mu}_\F(\G)=(\hat{\mu}_F(\G))_{F\in\F}$, and ${\bm\mu}_\F(f)=(\mu_F(f))_{F\in\F}$. Then, $\hat{\bm\mu}_\F(\G)$ is an unbiased, $\sqrt{N}$-consistent and asymptotically normal estimator of ${\bm\mu}_\F(f)$; i.e., $\E \hat{\bm\mu}_\F(\G) = {\bm\mu}_\F(f)$ and there exists a positive semi-definite ${\bm\Sigma}_\F(n,f)$ such that asymptotically in $N$:
\[
\sqrt{N}\big(\hat{\bm\mu}_\F(\G)-{\bm\mu}_\F(f)\big)
\to
\mathrm{Normal}\big(0,{\bm\Sigma}_\F(n,f)\big).
\]
Furthermore, for each $F,F'\in\F$,
\begin{equation}\label{cov}
\cov(\hat{\mu}_F(\G),\hat{\mu}_{F'}(\G))= \sum_{H\in\mathcal{H}_{FF'}\setminus \{F\sqcup F'\}}\omega_H(n;N)\big(\mu_H(f)-\mu_F(f)\mu_{F'}(f)\big),
\end{equation}
with $F\sqcup F'$ the disjoint union of $F,F'$, $\omega_H(n;N)\!=\!\frac1N\sum_{i=1}^N\frac{c_HX_H(K_{G_i})}{X_F(K_{G_i})X_{F'}(K_{G_i})}$, and $c_H$ is defined in \cref{LinearityLemma}.
\end{Theorem}
Crucial to the following is the covariance matrix ${\bm\Sigma}_\F(n,f)$---which can be obtained by taking the limit in $N$ in~\eqref{cov} for each $F,F'\in\F$---and which will enable the computation of confidence regions. Interestingly, its elicitation is more involved than for the study of large networks, where only a few terms dominate. We refer to the Appendix for the proof. The Appendix's proof relies on understanding the moments of single network counts, which have been studied extensively in the Erd\H{o}s--R\'enyi and the exchangeable cases, see for example~\cite{rucinski1988small,coulson2016poisson}.

\begin{Remark}[Non-random latent positions]
Because our proof build on the foundation of~\cite{rucinski1988small}, our results extend to the case where the latent $x_i$-s are not random; say fixed to some values. The only modification would be that the sum in~\eqref{cov} should be restricted to graphs $H$ formed by copies of $F$ and $F'$ overlapping over at least an edge (as opposed to overlapping over one node in~\eqref{cov}.) This also means that all our methods apply to that case, with the caveat that the resulting test will be conservative (as the variance is overestimated.) 
\end{Remark}

\begin{figure}[t]
\includegraphics[width=\textwidth]{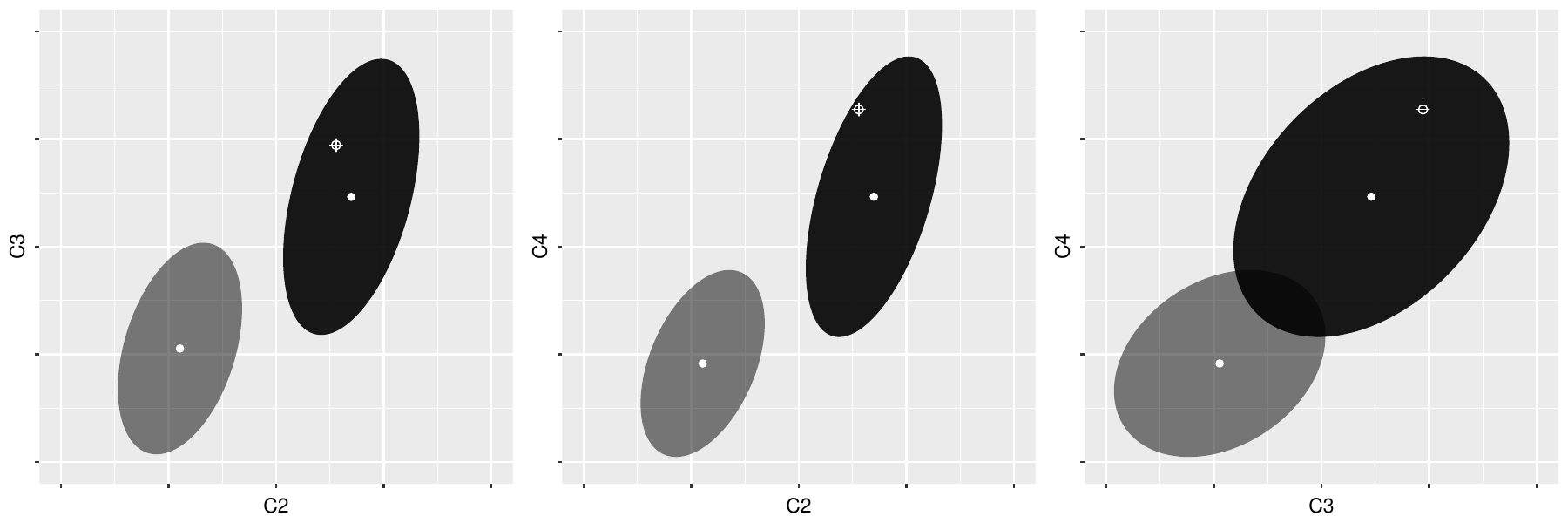}
\vspace{-2\baselineskip}
\caption{Testing for a kernel using $\F = \{C_2,C_3,C_4\}=\{\protect\Textedge\ \,,\protect\Texttriangle\ ,\protect\Textsquare\ \}$. The sample $\G$ is such that: $N=100$, $n_i$ is the $i$-th digit of $\pi$ plus $30$, $G_i$ is drawn from a kernel $f$ ($G_i\sim G_{n_i}(f)$). The estimate $\hat{\bm\mu}_\F(\G)$ is denoted by a white cross. Overlaid are the expected densities (white dots) and the confidence ellipse (shaded area) for two alternative kernels $f_a$ and $f_b$.
}
\label{Example1}
\end{figure}

\Cref{lyapu} enables testing against the null that all $G_i$ are drawn from a given kernel. Further, the Appendix contains a simulation experiment exploring convergence in small sample. To make this concrete, we consider an example in \cref{Example1}. There, we observe a graph sample $\G=(G_1,\dots,G_{100})$, and aim to compare it to two kernels $f_a$ (in black) and $f_b$ (in gray) using \cref{lyapu}; i.e., we assume that for $i\in[N]$, $G_i\sim G_{n_i}(f)$ and consider the null hypothesis $H_0: f=f_a$ and the alternative $H_1: f=f_b$. We draw as a white cross $\hat{\bm\mu}_\F(\G)$. The sizes of the networks in $\G$, the $n_i$, are non-random but not constant. We achieve this by using the sequence of digits of $\pi$.

First, since we have specified $f_a$ and $f_b$, we can evaluate both ${\bm\mu}_\F(f_a)$ and ${\bm\mu}_\F(f_b)$ and draw them on the figure (as white dots). Then, since $n=(n_i)_{i\leq N}$ is observed, we can compute ${\bm\Sigma}_\F(n,f_a)$ and ${\bm\Sigma}_\F(n,f_b)$ using \cref{lyapu}, which allows us to compute the confidence ellipse around ${\bm\mu}_\F(f_a)$ and ${\bm\mu}_\F(f_b)$ (in shaded black and gray respectively). Finally, since we know the limit distribution and covariance under the null, we can use ${\bm\Sigma}_\F(n,f_a)$ and ${\bm\mu}_\F(f_a)$ to compute a $p$-value using Mahalanobis distance; e.g., assuming ${\bm\Sigma}_\F(f_a)$ is full rank, the $p$-value is $1-F_{\chi^2_{|\F|}}\big((\hat{\bm\mu}_\F(\G)-{\bm\mu}_\F(f_a))'{\bm\Sigma}_\F(f_a)^{-1}(\hat{\bm\mu}_\F(\G)-{\bm\mu}_\F(f_a))\big)$, with $F_{\chi^2_{|\F|}}$ the cumulative distribution function of the $\chi^2$ distribution with $|\F|$ degrees of freedom.

\cref{Example1} is useful to understand the power of the proposed test. We see that since all the confidence ellipses do not all overlap, the power is larger than $.95$. However, if we were to use only \Texttriangle\ and \Textsquare, then the power would be less as the ellipses overlap. For larger sample sizes, the radii of both ellipses will be smaller, so that eventually the power tends to $1$ using any combination of subgraphs. In the Supplementary Material we explore the power of our test when the expected \Textedge\ count is held fixed, but the number of blocks in the null and true models differ.

In the following we consider the case where instead of testing against the null of a single kernel, we test against the null of a kernel class.

\section{The general case: Flexible sampling design}
Here we expand our results to cases where the observed networks may be generated from different kernels. Indeed, in many settings, the sampling mechanism may distort the structure of the underlying kernel; e.g., although the network connecting brain regions can be satisfactorily modeled by a blockmodel~\citep{koutra2013deltacon}, the proportion of nodes of each block may be different in different experimental settings, so that each observation is drawn from a different blockmodel.

In this practically important and conceptually challenging new setting, the proof techniques developed for \cref{lyapu} yield the following.
\begin{Theorem}\label{lyapu2}
Fix a set of graphs $\F$, a sequence of kernels $f=(f_i)_{i\in\mathbb{N}}$ and a sequence of integers $n=(n_i)_{i\in\mathbb{N}}$ such that $2\max_{F\in\F}|F|\leq\min_{i\in\mathbb{N}}n_i$. Let $\G=(G_i,\dots,G_N)$ be a network sample such that for all $i$, $G_i\sim G_{n_i}(f_i)$. Set $\hat{\bm\mu}_\F(\G)=N^{-1}\sum_{G\in\G}{(X_F(G)}/{X_F(K_G)})_{F\in\F}$ and ${\bm\mu}_\F(f;N)=N^{-1}\sum_{i=1}^N{\bm\mu}_\F(f_i)$. Then, asymptotically in $N$, and for some ${\bm\Sigma}_\F^\ast(n,f)$, we have
\[
\sqrt{N}\big(\hat{\bm\mu}_\F(\G)-{\bm\mu}_\F(f;N)\big) \to \mathrm{Normal}\big(0,{\bm\Sigma}_\F^\ast(n,f)\big).
\]
\end{Theorem}
Therefore, even in this much more flexible setting, we can recover the barycenter of the ${\bm\mu}_\F(f_i)$. However, the variance has now a more complex structure (see Appendix for details).

Following the intuition of our example of brain networks, and to make the usefulness of \cref{lyapu2} concrete, we introduce the \emph{flexible stochastic blockmodel} (FSBm).
\begin{Definition}[FSBm and embedding shape]
\label{FSBm}
For a symmetric matrix ${\bm B}\in[0,1]^{K\times K}$ we call $\D({\bm B})$ the set of all possible kernels with the same block structure as ${\bm B}$; i.e., recalling that $\mathcal{P}_K([0,1])$ is the set of partitions of $[0,1]$ in $K$ sets,
\[
\D({\bm B}) = \left\{f:\exists (P_1,\dots,P_K)\in\mathcal{P}_K([0,1]) \,s.t.\, \forall x\in P_s,y\in P_t, f(x,y)={\bm B}_{st}\right\}.
\]
For a set $\F$ of graphs, we call embedding shape the set
$
\mu_\F({\bm B}) = \left\{\mu_\F(f)\text{ for }f\in \D({\bm B})\right\}.
$
\end{Definition}
For instance, with ${\bm B}\in[0,1]^{2\times2}$ and $\F = \{\,\Textedge\ \,,\,\Texttriangle\ \}$, then:
\begin{multline*}
\mu_\F({\bm B}) = \left\{ \left(
	\pi^2B_{11}+2\pi(1-\pi)B_{12}+(1-\pi)^2B_{22},\right.\right.\\\left.\left.
	\pi^3B_{11}^3+3\pi^2(1-\pi)B_{11}B_{12}^2+3\pi(1-\pi)^2B_{22}B_{12}^2+(1-\pi)^3B_{22}^3
	\right)\,:\,\pi\in[0,1]\right\}.
\end{multline*}

The most direct way of using the FSBm is to test for all the $f_i$ being equal to any blockmodel instance in a class; i.e assume that all $G_i$-s are drawn from a kernel $f$ and test for the null $H_0: f\in\D({\bm B})$. This is achieved by using a composite hypothesis test, and our results allow us to produce confidence regions and $p$-values using the same tools as before.

\begin{figure}[t]
\centering
\includegraphics[width=\textwidth]{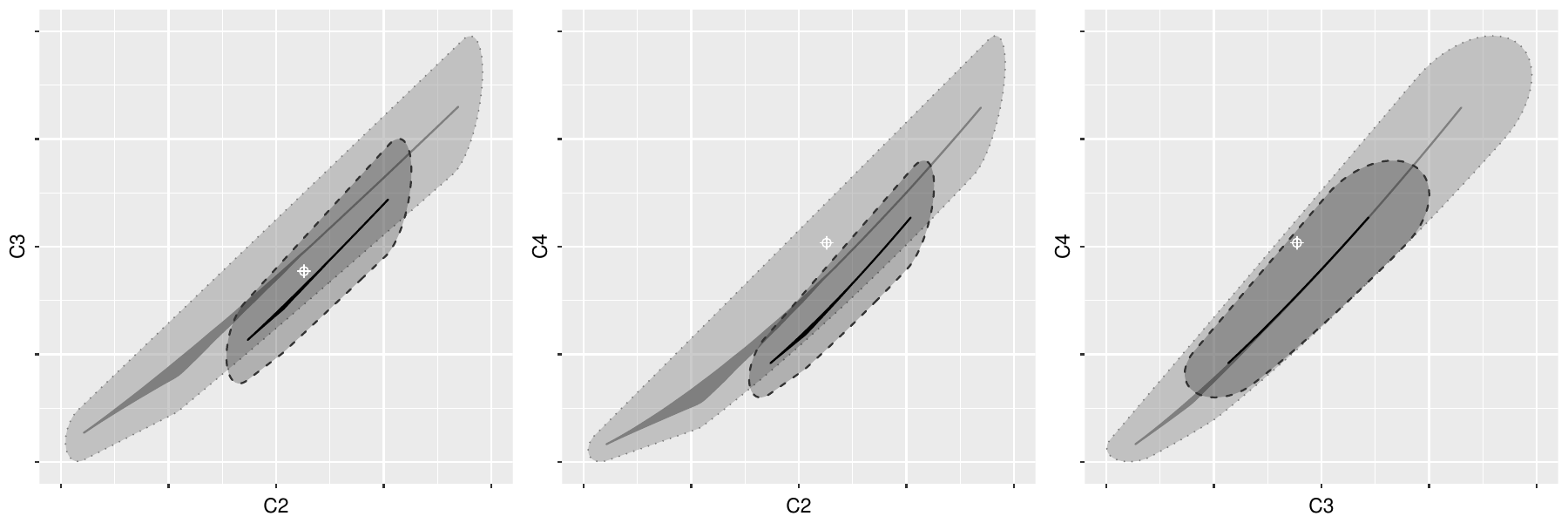}
\vspace{-2\baselineskip}
\caption{Testing for a FSBm class. The sample $\G$ is such that: $N=200$, $n_i$ is the $i$-th digit of $\pi$ plus $30$, $G_i$ is drawn from a kernel $f$ ($G_i\sim G_{n_i}(f)$). With $\F = \{C_2,C_3,C_4\}$, we estimate ${\hat{\bm\mu}}_\F(\G)$, and plot it as a white cross. Then, we draw in solid color the embedding shapes $\mu_\F({\bm B}_a)$ and $\mu_\F({\bm B}_b)$. In shaded color we draw the associated confidence regions; $p$-values can be obtained using the Mahalanobis distance associated with the closest point to ${\hat{\bm\mu}}_\F(\G)$ in $\mu_\F({\bm B}_a)$ and $\mu_\F({\bm B}_b)$.}
\label{Example2}
\end{figure}

We present such an example in \cref{Example2}. There, we observe $\G=(G_1,\dots,G_{200})$, and consider two FSBm classes generated from ${\bm B}_a$ (in gray) and ${\bm B}_b$ (in black). Then, we assume that all networks in the sample are drawn from a kernel $f$ and test for the null $H_0: f\in\D({\bm B}_a)$ and the alternative $H_1: f\in\D({\bm B}_b)$. We first represent ${\bm\mu}_\F(\G)$ as a white cross. Using \cref{FSBm}, we plot the em\-bedding shapes $\mu_\F({\bm B}_a)$ and $\mu_\F({\bm B}_b)$ in solid gray and black respectively. The confidence regions (in shaded gray with dotted contour and black with dashed contour) are the union of the confidence ellipses at all points in $\mu_\F({\bm B}_a)$ and $\mu_\F({\bm B}_b)$.

A more general use of \cref{lyapu2} is to test for all graphs in a sample being drawn from elements of a FSBm class; i.e assume that the $G_i$-s are drawn from the $f_i$-s and test for the null $H_0: \forall i\in[N], f_i\in\D({\bm B})$ for some ${\bm B}$. As before, we face a composite null, and we may compute the confidence region and the $p$-value by scanning all possible sequences $f$. This, however, is clearly computationally intractable. Nonetheless, the form of the variance and the structure of the FSBm allows us to propose conservative confidence regions and $p$-values that can be efficiently computed (we fully describe the method in the Supplementary Information).

\begin{figure}[t]
\centering
\includegraphics[width=\textwidth]{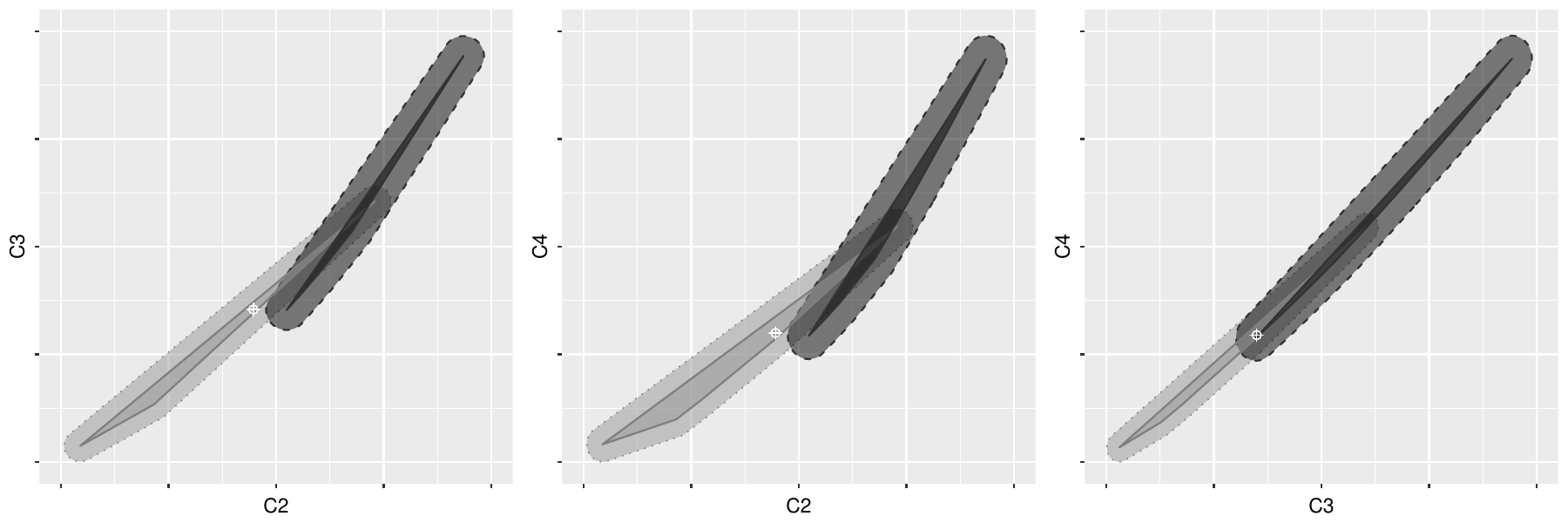}
\vspace{-2\baselineskip}
\caption{Testing for a full FSBm class. The sample $\G$ is such that: $N=10^3$, $n_i$ is the $i$-th digit of $\pi$ plus $50$, $G_i$ is drawn from a kernel $f_i$ ($G_i\sim G_{n_i}(f_i)$). With $\F = \{C_2,C_3,C_4\}$, we estimate ${\hat\mu}_\F(\G)$, and plot it as a white cross. Then, we draw in solid color the convex hulls of the embedding shapes $\mu_\F({\bm B}_a)$ and $\mu_\F({\bm B}_b)$. In shaded color we draw the associated confidence regions; approximate (and conservative) $p$-values can be obtained by determining the confidence level at which the observation ceases to be in the confidence region.
}
\label{Example3}
\end{figure}

We present an example in \cref{Example3}. There we observe $\G=(G_1,\dots,G_{10^3})$ and consider two FSBm classes generated by ${\bm B}_a$ and ${\bm B}_b$. We first plot ${\bm{\bm\mu}}_\F(\G)$ as a white cross. Then, using \cref{FSBm}, we plot the convex hull of the embedding shapes $\mu_\F({\bm B}_a)$ and $\mu_\F({\bm B}_b)$ (in solid gray and black respectively) wherein---by \cref{lyapu2}---${\bm\mu}_\F(f;N)$ must lie. Finally, we use a method described in the SI to produce the confidence region around each shape (in shaded hue).

\section{Are creative brains different from less creative ones?}
We now consider a sample of brain networks $\G = (G_1,\dots,G_{113})$~\citep{koutra2013deltacon}. This sample was produced as follows: diffusion MRI (dMRI) and structural MRI (sMRI) scans from $113$ individuals were collected over $2$ sessions from Beijing Normal University~\citep{corr}. Graphs were estimated using the {\tt ndmg} \citep{ndmg} pipeline. The dMRI scans were pre-processed for eddy currents using FSL's {\tt eddy-correct} \citep{fsl1}. FSL's ``standard'' linear registration pipeline was used to register the sMRI and dMRI images to the MNI152 atlas~\citep{fsl2,fsl3,mni152}. A tensor model is fit using DiPy~\citep{dipy} to obtain an estimated tensor at each voxel. A deterministic tractography algorithm is applied using DiPy's EuDX~\citep{eudx} to obtain streamlines, which indicate the voxels connected by an axonal fiber tract. Simple graphs are formed by first contracting voxels into graph vertices according to the Desikan parcellation~\citep{desikan} and then by placing an edge between two regions if a fiber tract is observed between any pair of voxel from each $70$ regions. Further, we have available a covariate $C = (c_1, \dots , c_{113})$ measuring subject's creativity, which is related to the person's performance on a series of creativity tasks.

To study this network sample and use the covariate $C$, we introduce a direct extension of our results to compare two network samples:
\begin{Corollary}[Two-sample test]\label{coro}
Fix a set of subgraphs $\F$ and two network samples $\G$ and $\G'$ generated respectively by the kernels $f$ and $f'$ and the network size sequences $n$ and $n'$. Then, as both $|\G|$ and $|\G'|$ tend to infinity, and if $\min(n,n')\geq2\max_{F\in\mathcal{F}}|F|$, we have that if $f=f'$, then
\[
\frac{\sqrt{|\G||\G'|}}{\sqrt{|\G|+|\G'|}}\big(\hat{\bm\mu}_\F(\G)-\hat{\bm\mu}_\F(\G')\big)\to\mathrm{Normal}\big(0,{\bm\Sigma}_\F(n,n',f)\big),
\]
and, as~\eqref{cov} applies, we can produce an unbiassed $\sqrt{|\G|+|\G'|}$-consistent and asymptotically normal estimator of $\mathrm{vec}\big({\bm\Sigma}_\F(n,n',f)\big)$, this without specifying or estimating $f$.
\end{Corollary}
In the following we choose to work with $\F = \{\Textedge\ \,,\Texttriangle\ ,\Textsquare\ \}$. We chose these for several reasons: in a simulation experiment (see Supplementary Material) we find that using these counts leads to power surpassing the best available alternative in the literature~(\citet{tang2016nonpara}, \citet{tang2016semipara}); the associated counts correspond to spectral moments of the underlying kernel~\citep{maugis2017topo}; \Textedge\ and \Texttriangle are staples of graph analysis (related to conductance, clustering, transitivity, ...) and are key network summaries throughout the literature; \Textsquare\ is one of the smallest graphs that is larger than \Textedge\ and \Texttriangle but cannot be built from copies of \Textedge\ and \Texttriangle, which reduces correlation between counts, and therefore improve power of the tests; all three are small, and there is extensive and available software to efficiently count them (e.g.,~\citet{maugis2017fast,jha2015path,pinar2016escape}.)

We note that our estimator of ${\bm\Sigma}_\F(n,n',f)$ is entrywise Normal, not Wishart as in the classical setting. Thus, we have no guarantee that $\widehat{\bm\Sigma}_\F(n,n',f)$ is positive definite, and cannot use the Hotelling's $T$-squared distribution to compute $p$-values. If the estimate is positive definite, we recommend ignoring the variations in $\widehat{\bm\Sigma}_\F(n,n',f)$ and using the $\chi_{|\F|}^2$ distribution. If the estimate is not positive definite, we recommend using the marginals.

Before analyzing $\G$ using our results, we make the following test: we subsample uniformly at random and without replacement from $\G$, yielding $\G_1$ and $\G_2$ such that $\G_1\cup\G_2=\G$, and use \cref{coro} to test for $\G_1$ and $\G_2$ being drawn from the same kernel $f$. Unless $\G$ presents characteristics that cannot be explained by our results, $\G_1$ and $\G_2$ should be indistinguishable, and we expect to see $p$-values that are uniformly distributed in $[0,1]$.

We perform this experiment $100$ times, and obtain a sample of $p$-values for which we fail to reject the null of a uniform distribution using the Kolmogorov--Smirnov test ($D = 0.09$, $p$-value $=0.3$). For this test we use $\F=\{\Textsquare\ \}$ because of the small sample ($|\G_1|+|\G_2|=113$) size and a very high level of correlation; otherwise the correlation between the counts is so high that the correlation matrix appears singular.

We now use $C$ to split $\G$ in two samples. To do so, we choose to build a first subsample $\G_1$ containing the less creative, and a second subsample $\G_2$ containing the more creative. More precisely, for a quantile $q$ and denoting $Q_C$ the empirical quantile function of $C$:
\[
\G_1^q = \{G_i\in\G\,:\, c_i\leq Q_C(q)\}
\quad\&\quad
\G_2^q = \{G_i\in\G\,:\, c_i>Q_C(1-q)\}.
\]
Interestingly, for $q=0.5$ and $q=0.4$, we fail to reject the null that the networks in $\G_1^q$ and $\G_2^q$ come from the same kernel (see \cref{table} for the $p$-values). However, for $q=0.3$ we can reject the null of the same kernel at the $5\%$ confidence level using~\Texttriangle~or~\Textsquare\,, but not~\Textedge\,.

\begin{table}\centering
\renewcommand{\arraystretch}{.8}
\begin{tabular}{cccc}
\toprule
\multirow{2}{*}{Quantile ($q$)}
&\multicolumn{3}{c}{$p$-value} \\
\cmidrule(r){2-4}
& \multicolumn{1}{c}{\Textedge}
& \multicolumn{1}{c}{\Texttriangle}
& \multicolumn{1}{c}{\Textsquare} \\
\midrule
0.5 & 0.126 & 0.110 & 0.115\\
0.4 & 0.077 & 0.051 & 0.050\\
0.3 & 0.062 & 0.042 & 0.040\\
0.2 & 0.014 & 0.011 & 0.012\\
0.1 & 0.046 & 0.047 & 0.061\\
\bottomrule
\end{tabular}
\caption{Testing for differences between $\G_1^q$ and $\G_2^q$. For each $F\in\{\protect\Textedge\ \,,\,\protect\Texttriangle\ \,,\,\protect\Textsquare\ \,\}$ and $q\in\{0.1, 0.2, \dots, 0.5\}$ we produce the $p$-value for the null $H_0:\mu_F(\G_1^q)=\mu_F(\G_2^q)$. The $p$-values increase with $q$, except for $q=0.1$, in which case $|\G_1^q|$ and $|\G_2^q|$ are too small for the test to be significant.}\label{table}
\renewcommand{\arraystretch}{1}
\end{table}

Thus, we observe that individuals with a very high level of creativity present significantly more~\Texttriangle~and~\Textsquare~than those with a very low level of creativity. To further confirm this discovery, we undertake the following experiment: we estimate kernels (through the random-dot-product framework~\citep{athreya2018survey}) over the samples $\G$, $\G_1^{0.2}$, and $\G_2^{0.2}$, that we call $f$, $f_1$, and $f_2$ respectively; then, for several $\gamma\in[0,1]$, we consider the power of \cref{coro}'s test when the samples compared are of the same cardinality as $\G_1^{0.2}$, and $\G_2^{0.2}$, but drawn i.i.d. from $G_{70}\big(\gamma f_1 + (1-\gamma)f\big)$ and $G_{70}\big(\gamma f_2 + (1-\gamma)f\big)$ respectively; we find that the proposed test is very powerful, as powerful than a semi-parametric test relying on the random-dot-product structure of the model~\citep{tang2016semipara}. This comforts the claim that the difference in sample is only observed for $q\geq .3$, and not for more central quantiles.

We now aim to understand whether the added~\Texttriangle~and~\Textsquare~arise from a few edges completing partially present shapes or from fully new~\Texttriangle~and~\Textsquare. To do so, we first observe that if $G\sim G_n(f)$, then $\overline G\sim G_n(1-f)$, where $\overline G$ is the complement graph of $G$. Therefore, we may use our tests on $\overline G$, which can be understood as estimating ${\bm\mu}_\F(1-f)$ instead of ${\bm\mu}_\F(f)$ to compare network samples.

Then, using the $\overline\G_i^q = \{\overline G\,:\, G\in \G_i^q\}$, we can test whether there are significantly more fully absent subgraphs in $\G_1^q$ compared to $\G_2^q$. There, we find we cannot reject this null; i.e., we cannot reject the null of the networks in $\overline\G_1^q$ and $\overline\G_2^q$ coming from the same kernel for $q\geq0.3$. Therefore, we conclude that the added~\Texttriangle~and~\Textsquare~in the highly creative arise from a few edges completing partially present~\Texttriangle~and~\Textsquare.

One key conclusion of our analysis is that we must use~\cref{lyapu2} to study $\G$. Indeed, we have just shown that all networks generating $\G$ cannot come from the same kernel. This allows us to write, that using the full sample $\G$, a centered and consistent estimate of the average density of $\protect\Textedge\ \,,\,\protect\Texttriangle\ \,,\,\protect\Textsquare\ \,$ within human brains is $(0.41,0.13,0.06)$. To conclude, we remark that since all the $G_i$s in $\G$ have the same number of nodes, one could produce a similar analysis using---instead of our results---the ${\bm\mu}_\F(G_i)$ as if they formed an independent and identically distributed sample. However, as we have just shown, the $G_i$ are not identically distributed, and therefore such an analysis could lead to spurious conclusions.
\section{Discussion}
We provide the tools to perform statistical inference on a network sample using subgraph counts. Our two main results provide consistency and asymptotic normality of subgraph counts under very flexible conditions. Using these results, we show that subgraph counts are powerful statistics to test whether network samples come from a specified distribution, a specified model, or from the same model.

The key insight we provide is that statistical inference methods paralleling classical ones for standard samples may be obtained for network samples. From this perspective, our results may be seen as providing an analog of a multivariate $t$-test for network samples. However, going beyond what our results directly imply, we expect that parallels to ANOVA, model selection, model ranking, and goodness of fit may be obtained for network samples using our proof techniques.

However, since our tests are analog to the $t$-test, they rely on global summaries. Therefore, while we can reject the null of two network samples being drawn from the same model, our approach is unable to locate where in the network the difference is stemming from.
\appendix
\makeatletter
\renewcommand{\@seccntformat}[1]{
	APPENDIX~{\csname the#1\endcsname}:\hspace*{1em}}
\makeatother
\section{Properties of subgraph counts}
In the following we formalize certain notions we use loosely in the main body (especially the notion of copy and the sets $\mathcal{H}_{FF'}$, as well as the constants $c_H$), prove all our results, and provide more details on the numerical examples we present.

We start by proving our first lemma, establishing the linearity of subgraph counts. We first define $\mathcal{H}_{F_1F_2}$ and $c_H$, generalizing definitions given in~\cite{rucinski1988small}.

\begin{Definition}[Overlapping copies]\label{overlapping-copies}
For two graphs $F_1$ and $F_2$ we write $\mathcal{H}_{F_1F_2}$ the set of unlabeled graphs that can be formed by two copies of $F_1$ and $F_2$, and $c_H$ the number of ways a given $H\in\mathcal{H}_{F_1F_2}$ can be built from copies of $F_1$ and $F_2$:
\[\begin{dcases}
\mathcal{H}_{F_1F_2} &=
	\left\{H\subset K_{|F_1|+|F_2|}\,:
		\parbox{.42\textwidth}{\centering
			$\exists F'_1,F'_2\subset K_{|F_1|+|F_2|}$ s.t.,\\[.1\baselineskip]
			$F_1'\equiv F_1, F'_2\equiv F_2$,
			and $H = F_1'\cup F_2'$}\right\}/\equiv\\[.3\baselineskip]
c_H &= \#\left\{\left(F_1',F_2'\right)\subset H\,:
	\parbox{.33\textwidth}{\centering
		$F_1'\equiv F_1, F_2'\equiv F_2$\\[.1\baselineskip]
		and $H= F_1'\cup F_2'$}\right\}.
\end{dcases}\]
Finally, call $\mathcal{H}_{F_1F_2}^\ast$ the set $\mathcal{H}_{F_1F_2}$ removed of $F_1\sqcup F_2$, the vertex disjoint union of $F_1$ and $F_2$.
\end{Definition}

Note that $\mathcal{H}_{F_1F_2}$ is defined as a quotient set, sometimes called quotient space, through the equivalence relation $\equiv$. In the following we identify the equivalence classes in $\mathcal{H}_{F_1F_2}$ with any of their element, and therefore will treat the elements of $\mathcal{H}_{F_1F_2}$ as graphs.

\begin{Lemma}[Copies pairwise interaction]\label{copies-product}
Fix three graphs $F_1,F_2$ and $G$. Then,
\[
X_{F_1}(G)X_{F_2}(G) =
\sum_{H\in\mathcal{H}_{F_1F_2}}c_HX_H(G).
\]
\end{Lemma}
\begin{proof}
We start by writing:
\begin{equation}
\label{copies-product-1}
X_{F_1}(G)X_{F_2}(G)
= \sum_{F_1'\subset G}1_{\{F_1'\equiv F_1\}}
\sum_{F_2'\subset G}1_{\{F_2'\equiv F_2\}}
= \sum_{\substack{F_1'\subset G\\F_2'\subset G}}1_{\{F_1'\equiv F_1\}}1_{\{F_2'\equiv F_2\}}.
\end{equation}
Now, from \cref{overlapping-copies}, we first note that by construction of $\mathcal{H}_{F_1F_2}$, for each pair $F_1',F_2'$ in the sum, $1_{\{F_1'\equiv F_1\}}1_{\{F_2'\equiv F_2\}}=1$ if and only if there exists $H\in\mathcal{H}_{F_1F_2}$ such that $F_1'\cup F_2' \equiv H$. Therefore we can reindex the sum in~\eqref{copies-product-1} as follows:
\begin{equation*}
X_{F_1}(G)X_{F_2}(G)
= \sum_{F_1',F_2'\subset G}1_{\{\exists H\in\mathcal{H}_{F_1F_2}\,:\, F_1'\cup F_2'\equiv H\}}
= \sum_{H\in\mathcal{H}_{F_1F_2}}\ \sum_{F_1',F_2'\subset G}1_{\{F_1'\cup F_2'\equiv H\}}.
\end{equation*}
We now note that by definition of $c_H$, for each copy of $H$ in $G$, there will be $c_H$ pairs $(F_1',F_2')$ of copies of $F_1$ and $F_2$ in $G$ such that $F_1'\cup F_2' = H$. Therefore we can simplify the sum above to obtain:
\begin{equation*}
X_{F_1}(G)X_{F_2}(G)
= \sum_{H\in\mathcal{H}_{F_1F_2}} c_H\sum_{H'\subset G}1_{\{H'\equiv H\}}
= \sum_{H\in\mathcal{H}_{F_1F_2}} c_HX_H(G),
\end{equation*}
yielding the desired result.
\end{proof}

With these tools in hand, we compute the first two moments of $X_F(G)$ when $G\sim G(n,f)$.

\begin{Proposition}\label{first-moms}
Fix two graphs $F$ and $F'$ and a random graph $G\sim G(|G|,f)$ such that $|F|+|F'|\leq |G|$. Then, we have that
\begin{align*}
\E X_F(G)&= X_F(K_G)\mu_F(f)\\
\cov(X_F(G),X_{F'}(G))&= \sum_{H\in\mathcal{H}_{FF'}^\ast}c_HX_H(K_G)\big(\mu_H(f)-\mu_F(f)\mu_{F'}(f)\big).
\end{align*}
\end{Proposition}
\begin{proof}
We prove each statement in succession. To begin, call $F_1,\dots,F_m$ the $m=X_F(K_G)$ copies of $F$ in $K_G$ ($m>0$ because $|F|\leq |G|$). Then,
$
X_F(G) = \sum_{i\in[m]}1_{\{F_i\subset G\}}, %
$
and by linearity of the expectation,
\begin{equation}\label{mean-1}
\E X_F(G) = \sum_{i\in[m]}\E1_{\{F_i\subset G\}}.
\end{equation}
Let us fix $i\in[m]$ and consider $\E1_{\{F_i\subset G\}}$. To do so we will use the law of total probability:
\begin{align*}
\E1_{\{F_i\subset G\}}
= \E \prod_{pq\in F_i} 1_{\{pq\in G\}}
&= \E\left[\E \left[\,\prod_{pq\in F_i} 1_{\{pq\in G\}}\,\Big\vert\,\{x_l\}_{l\in F_i}\,\right]\right]\\
&= \E\left[\,\prod_{pq\in F_i} f(x_p,x_q)\,\right]=\mu_F(f).
\end{align*}
Therefore, $\E1_{\{F_i\subset G\}}$ does not depend on $i$, and resuming from~\eqref{mean-1}, we obtain
\begin{equation}\label{mean-2}
\E X_F(G) = \sum_{i\in[m]}\mu_F(f)=m\mu_F(f)
\end{equation}
which is the desired result.

We now turn to the variance. Call $F_1',\dots,F_{m'}'$ the $m'=X_{F'}(K_G)$ copies of $F'$ in $K_G$ ($m'>0$ because $|F'|\leq |G|$). We first write that
\begin{align}
\nonumber
\cov(X_F(G),X_{F'}(G))
&=\E \left[
\left(\sum_{i\in[m]}\left\{1_{\{F_i\subset G\}}-\mu_F(f)\right\}\right)\left(\sum_{i\in[m']}\left\{1_{\{F_i'\subset G\}}-\mu_{F'}(f)\right\}\right)\right]
\\
\label{var-1}
&= \sum_{i\in[m],i\in[m']}\E \left[\big(1_{\{F_i\subset G\}}-\mu_F(f)\big)\big(1_{\{F_{i'}'\subset G\}}-\mu_{F'}(f)\big)\right].
\end{align}
Now, we observe that if $F_i$ is disjoint form $F_{i'}'$ (i.e., $F_i\cap F_{i'}'=\emptyset$ which is possible because $|F|+|F'|\leq |G|$) then $1_{\{F_i\subset G\}}$ is independent from $1_{\{F_{i'}'\subset G\}}$ and we have that
\begin{align*}
\E \big[\big(1_{\{F_i\subset G\}}-\mu_F(f)\big)\big(1_{\{F_{i'}'\subset G\}}-\mu_{F'}(f)\big)\big]
&= \E \big(1_{\{F_i\subset G\}}-\mu_F(f)\big)\E\big(1_{\{F_{i'}'\subset G\}}-\mu_{F'}(f)\big)\\
&=\big(\mu_F(f)-\mu_F(f)\big)\big(\mu_{F'}(f)-\mu_{F'}(f)\big)
=0.
\end{align*}
Therefore resuming from~\eqref{var-1} we obtain that
\begin{align*}
\cov(X_F(G),X_{F'}(G))
&=\sum_{\substack{i\in[m],j\in[m']\\F_i\cap F_[i']'\neq\emptyset}}\E \left[\big(1_{\{F_i\subset G\}}-\mu_F(f)\big)\big(1_{\{F_{i'}'\subset G\}}-\mu_{F'}(f)\big)\right]\\
&=\sum_{\substack{i\in[m],j\in[m']\\F_i\cap F_[i']'\neq\emptyset}}\left\{\E \left[1_{\{F_i\subset G\}}1_{\{F_{i'}'\subset G\}}\right]-\mu_F(f)\mu_{F'}(f)\right\}\\
&=\sum_{\substack{i\in[m],j\in[m']\\F_i\cap F_[i']'\neq\emptyset}}\left\{\E \left[1_{\{F_i\cup F_{i'}'\subset G\}}\right]-\mu_F(f)\mu_{F'}(f)\right\}.
\end{align*}
Then, by the exact same transformation we used in \cref{copies-product}, we know that each $F_i\cup F_{i'}'$ is in $\mathcal{H}_{FF'}^\ast$ (because they are not disjoint), and furthermore, for each copy of $H$ in $G$ there are $c_H$ pairs $F_u$ and $F_v'$ such that $H = F_v\cup F_v'$, so that
\begin{align*}
\cov(X_F(G),X_{F'}(G))
&= \sum_{H\in\mathcal{H}_{FF'}^\ast}c_H\sum_{H'\subset G}\big\{\E(1_{\{H'\equiv H\}})-\mu_F(f)\mu_{F'}(f)\big\}\\
&= \sum_{H\in\mathcal{H}_{FF'}^\ast}c_H\big(\E X_H(G)-X_H(K_G)\mu_F(f)\mu_{F'}(f)\big)\\
&= \sum_{H\in\mathcal{H}_{FF'}^\ast}c_HX_H(K_G)\big(\mu_H(f)-\mu_F(f)\mu_{F'}(f)\big),
\end{align*}
which is the desired result.
\end{proof}

We now turn to the proofs of \cref{lyapu,lyapu2}. As the second generalizes the first, it is sufficient to prove the second.
\begin{proof}
We obtain the result by a joint application of the Lindeberg-Feller central limit theorem and the Cramer-Wold device. To do so, we fix ${\bm a}\in\mathbb{R}^{|\F|}$ and compute the variance of our estimator projected along ${\bm a}$.

\noindent{\bf Computing the variance: }
First recall that %
$
\hat{\bm{\mu}}_\F(\G) = N^{-1}\sum_{i\in[N]}\left(\frac{X_F(G_i)}{X_F(K_{G_i})}\right)_{F\in\F}.
$ 
Therefore, denoting ``$\cdot$'' the inner product and taking the expectation over $\G$, let
\[
s_N^2 = \var\left({\bm a}\cdot{\hat{\bm\mu}}_\F(\G)\right).
\]
Then, using the independence of the $G_i$s and the bi-linearity of the covariance, we have
\begin{align*}
s_N^2
= N^{-2}\sum_{i\in[N]}\var\left({\bm a}\cdot\left(\frac{X_F(G_i)}{X_F(K_{G_i})}\right)_{F\in\F}\right)
= N^{-2}\sum_{i\in[N]}\sum_{F,F'\in\F}\frac{{\bm a}_F{\bm a}_{F'}\cov\left(X_F(G_i),X_{F'}(G_i)\right)}{X_F(K_{G_i})X_{F'}(K_{G_i})}.
\end{align*}
To proceed, recall that for each $i\leq N$, $G_i\sim G(n_i,f_i)$. Then, we may use \cref{first-moms} to obtain
\[
s_N^2 = N^{-2}\sum_{i\in[N]}\sum_{F,F'\in\F}
\frac{{\bm a}_F{\bm a}_{F'}\sum_{H\in\mathcal{H}_{FF'}^\ast}c_HX_H(K_{G_i})\big(\mu_H(f_i)-\mu_F(f_i)\mu_{F'}(f_i)\big)}{X_F(K_{G_i})X_{F'}(K_{G_i})}.
\]
Because all sums are finite, we can reorder the summations, leading to
\begin{align}
\nonumber
s_N^2 &= N^{-1}\sum_{F,F'\in\F}{\bm a}_F{\bm a}_{F'}
\sum_{H\in\mathcal{H}_{FF'}^\ast}\left(N^{-1}\sum_{i\in[N]}\frac{c_HX_H(K_{G_i})}{X_F(K_{G_i})X_{F'}(K_{G_i})}\big(\mu_H(f_i)-\mu_F(f_i)\mu_{F'}(f_i)\big)\right)\\
\label{lindeberg-1}
&=N^{-1}\sum_{F,F'\in\F}{\bm a}_F{\bm a}_{F'}\sum_{H\in\mathcal{H}_{FF'}^\ast}\omega_H(n,f;N),
\end{align}
where
$
\omega_H(n,f;N) = N^{-1}\sum_{i\in[N]}\frac{c_HX_H(K_{G_i})}{X_F(K_{G_i})X_{F'}(K_{G_i})}\big(\mu_H(f_i)-\mu_F(f_i)\mu_{F'}(f_i)\big).
$

\noindent{\bf Convergence of $\omega_H(n,f;N)$: }
To proceed we must show that $\omega_H(n,f;N)$ converges to a limit as $N$ diverges. We will achieve this by showing that $\omega_H(n,f;N)$ is Cauchy. To do so, observe that
\begin{align*}
|\omega_H&(n,f;N+1)-\omega_H(n,f;N)|\\
&= \left|(N+1)^{-1}\left(N\omega_H(n,f;N)+\frac{c_HX_H(K_{G_{N+1}})}{X_F(K_{G_{N+1}})X_{F'}(K_{G_{N+1}})}\right.\right.\qquad\quad\\
&\pushright{\left.\left.
\vphantom{\frac{c_HX_H(K_{G_{N+1}})}{X_F(K_{G_{N+1}})X_{F'}(K_{G_{N+1}})}}
\big(\mu_H(f_{N+1})-\mu_F(f_{N+1})\mu_{F'}(f_{N+1})\big)\right)-\omega_H(n,f;N)\right|}\\
&= \left|(N+1)^{-1}\frac{c_HX_H(K_{G_{N+1}})}{X_F(K_{G_{N+1}})X_{F'}(K_{G_{N+1}})}\big(\mu_H(f_{N+1})-\mu_F(f_{N+1})\mu_{F'}(f_{N+1})\big)-\frac{1}{N+1}\right|.\!
\end{align*}
Then, we recall that $X_F(K_G)=\binom{|G|}{|F|}\mathrm{aut}(F)$ (see for instance~\cite{bollobas2009metric}), where $\mathrm{aut}(F)$ is the number of automorphisms of $F$ (the number of bijections from the vertex set of $F$ to itself that preserve adjacency.) Then, as $|H|<|F|+|F'|$, we have
\[
\frac{X_H(K_{G_{N+1}})}{X_F(K_{G_{N+1}})X_{F'}(K_{G_{N+1}})} =
\frac{\mathrm{aut}(H)}{\mathrm{aut}(F)\mathrm{aut}(F')}\frac{\binom{n_{N+1}}{|H|}}{\binom{n_{N+1}}{|F|}\binom{n_{N+1}}{|F'|}}\leq \frac{\mathrm{aut}(H)}{\mathrm{aut}(F)\mathrm{aut}(F')}.
\]
As furthermore $f_{N+1}$ is bounded by $1$, we have
$
|\big(\mu_H(f_{N+1})-\mu_F(f_{N+1})\mu_{F'}(f_{N+1})\big)|\leq1, %
$
leading to $|\omega_H(n,f;N+1)-\omega_H(n,f,N)|\leq \frac{C}{N+1}$ for $C\!=\!1\!+\!c_H\mathrm{aut}(H)/{\mathrm{aut}(F)\mathrm{aut}(F')}$. Thus, the sequence $\omega_H(n,f;N)$ is Cauchy, and we may call $\omega_H(n,f)$ its limit; i.e.,
\[
\lim_{N\to\infty}\omega_H(n,f;N) = \omega_H(n,f).
\]

Then, resuming from~\eqref{lindeberg-1}, and writing ${\bm\Sigma}_\F$ the matrix indexed by $\F$ such that
\[
({\bm\Sigma}_\F(n,f))_{FF'}=\sum_{H\in\mathcal{H}_{FF'}^\ast}\omega_H(n,f),
\]
we have $\lim_{N\to\infty} Ns_N^2 = {\bm a}^\top{\bm\Sigma}_\F(n,f){\bm a}$.

\noindent{\bf Satisfying the Lindeberg-Feller condition: }
To invoke the Lindeberg-Feller central limit theorem, we must show that our sequence verifies the so called Lindeberg-Feller condition. Recall that the sequence under study is
\[
Y_i:={\bm a}\cdot (X_F(G_i)/X_F(K_{G_i}))_{F\in\F}-{\bm a}\cdot{\bm\mu}_\F(f_i),
\]
and that the variance of the partial sum is $N^2s_N^2$. Therefore, the Lindeberg-Feller condition we need to satisfy is the following:
\begin{equation}\label{lindeberg-2}
\forall \epsilon>0 \lim_{N\to\infty}\frac{1}{N^2s_N^2}\sum_{i\in[N]}
\E\left[Y_i1_{\{|Y_i|>\epsilon Ns_N\}}\right]=0.
\end{equation}
To verify the condition we first fix $\epsilon>0$. Then, observe that since for each $i$ and $F$ we have $X_F(G_i)/X_F(K_{G_i})\leq1$ and $\mu_F(f_i)\leq1$, we have that $|Y_i|\leq 2\|{\bm a}\|_1$ by the triangle inequality. Therefore, as $Ns_N\to\infty$ as $N$ grows, we may fix an $N$ such that for all $N'>N$ we have $\|{\bm a}\|_1\leq\epsilon Ns_N$. In this setting, the sum in~\eqref{lindeberg-2} is equal to zero for all $N'>N$, and the condition is verified. Therefore, we have that
\[
\frac{1}{Ns_N}\sum_{i\in[N]}Y_i\to\mathrm{Normal}(0,1).
\]

Reverting to the notation of the statement of the Theorem, we have that for any ${\bm a}$
\[
\sqrt{N}\big({\bm a}\cdot {\hat{\bm{\mu}}}_\F(\G)-{\bm a}\cdot {\bm{\mu}}_\F(f)\big)\to\mathrm{Normal}(0,{\bm a}^\top{\bm\Sigma}_\F(n,f){\bm a}),
\]
which is sufficient to obtain the claimed result for the limit in distribution.
\end{proof}

We now turn to the proof of \cref{coro}.
\begin{proof}
The result follows almost immediately from \cref{lyapu} and Slutsky's theorem.

First, since $|\G|$ and $|\G'|$ tend to infinity and both $n$ and $n'$ are large enough, we have
\[\begin{cases}
\sqrt{|\G|}\big(\hat{\bm\mu}_\F(\G)-{\bm\mu}_\F(f)\big) \to \mathrm{Normal}\big(0,{\bm\Sigma}_\F(n,f)\big),\\
\sqrt{|\G'|}\big(\hat{\bm\mu}_\F(\G')-{\bm\mu}_\F(f)\big) \to \mathrm{Normal}\big(0,{\bm\Sigma}_\F(n',f)\big).
\end{cases}\]
As both samples are independent, linear combinations are also multivariate Gaussian. Therefore, multiplying the first line by $\sqrt{|\G'|}/{\sqrt{|\G|+|\G'|}}$, the second by $\sqrt{|\G|}/{\sqrt{|\G|+|\G'|}}$ (both ratios being in $(0,1)$, the limit in distribution is unaffected), and taking the difference, we obtain, with ${\bm\Sigma}_\F(n,n',f)=\lim_{|\G|,|\G'|\to\infty}\left\{\frac{|\G'|}{|\G|+|\G'|}{\bm\Sigma}_\F(n,f)\!+\!\frac{|\G|}{|\G|+|\G'|}{\bm\Sigma}_\F(n',f)\right\}$
\begin{align*}
\frac{\sqrt{|\G||\G'|}}{\sqrt{|\G|+|\G'|}}\big(\hat{\bm\mu}_\F(\G)-\hat{\bm\mu}_\F(\G)\big)
\to \mathrm{Normal}\big(0,{\bm\Sigma}_\F(n,n',f)\big),
\end{align*}
which is the desired limit in distribution.

To obtain the consistent estimator of ${\bm\Sigma}_\F(n,n',f)$ we first recall that for $F,F'\in\F$
\[
\big({\bm\Sigma}_\F(n,f)\big)_{FF'} =
\sum_{H\in\mathcal{H}_{FF'}^\ast}\omega_H(n)\big(\mu_H(f)-\mu_F(f)\mu_{F'}(f)\big).
\]
Then, with
$
\omega_H(n,n') = \lim_{|\G|,|\G'|\to\infty}\left(\frac{|\G'|}{|\G|+|\G'|}\omega_H(n)+\frac{|\G|}{|\G|+|\G'|}\omega_H(n')\right),
$
we have
\[
\big({\bm\Sigma}_\F(n,n',f)\big)_{FF'} =
\sum_{H\in\mathcal{H}_{FF'}^\ast}\omega_H(n,n')\big(\mu_H(f)-\mu_F(f)\mu_{F'}(f)\big).
\]
There observe that:
\begin{itemize}[leftmargin=*]
\item[--] As $|\G|$ and $|\G'|$ grow we have that
\begin{align*}
\omega_H(n,n';|\G|,|\G'|)
=\frac{|\G'|}{|\G|+|\G'|}\omega_H(n;|\G|)+\frac{|\G|}{|\G|+|\G'|}\omega_H(n',|\G'|)
\to \omega_H(n,n').
\end{align*}
\item[--] All the $\mu_H(f)$ in $H\in\mathcal{H}_{FF'}$ may be estimated by
$
\hat\mu_H(\G\cup\G') = \frac{1}{|\G\cup\G'|}\sum_{G\in\G\cup\G'}\frac{X_H(G)}{X_H(K_G)}.
$
Furthermore, both $\mu_F(f)$ and $\mu_{F'}(f)$ may be estimated in the same way.
\end{itemize}
Then, by a direct application of the Slutsky's theorem, we have that
\[
\big({\hat{\bm\Sigma}}_\F(n,n',f)\big)_{FF'} \!=\!\!
\sum_{H\in\mathcal{H}_{FF'}^\ast}\!\!\omega_H(n,n';|\G|,|\G'|)\big(\hat\mu_H(\G\cup\G')-\hat\mu_F(\G\cup\G')\hat\mu_{F'}(\G\cup\G')\big)
\]
is an asymptotically normal estimator of $\big({\bm\Sigma}_\F(n,n',f)\big)_{FF'}$ converging at rate $\sqrt{|\G|+|\G'|}$, which yields the claimed result.
\end{proof}

\section{Finite sample experiments}
Before we proceed, we consider a simulation experiment to determine how large the sample size $N$ must be for the asymptotic limit to be a satisfactory approximation of the statistic's distribution. We present our result in \cref{here:Experiment}. There, we observe that fairly small $N$, on the order of 100 even with small $n_i$-s, may be sufficient. Furthermore, results from~\cite{BickelLevina2012} suggest that larger networks would make this convergence faster.

\begin{figure}
\centering
\includegraphics[width=.8\textwidth]{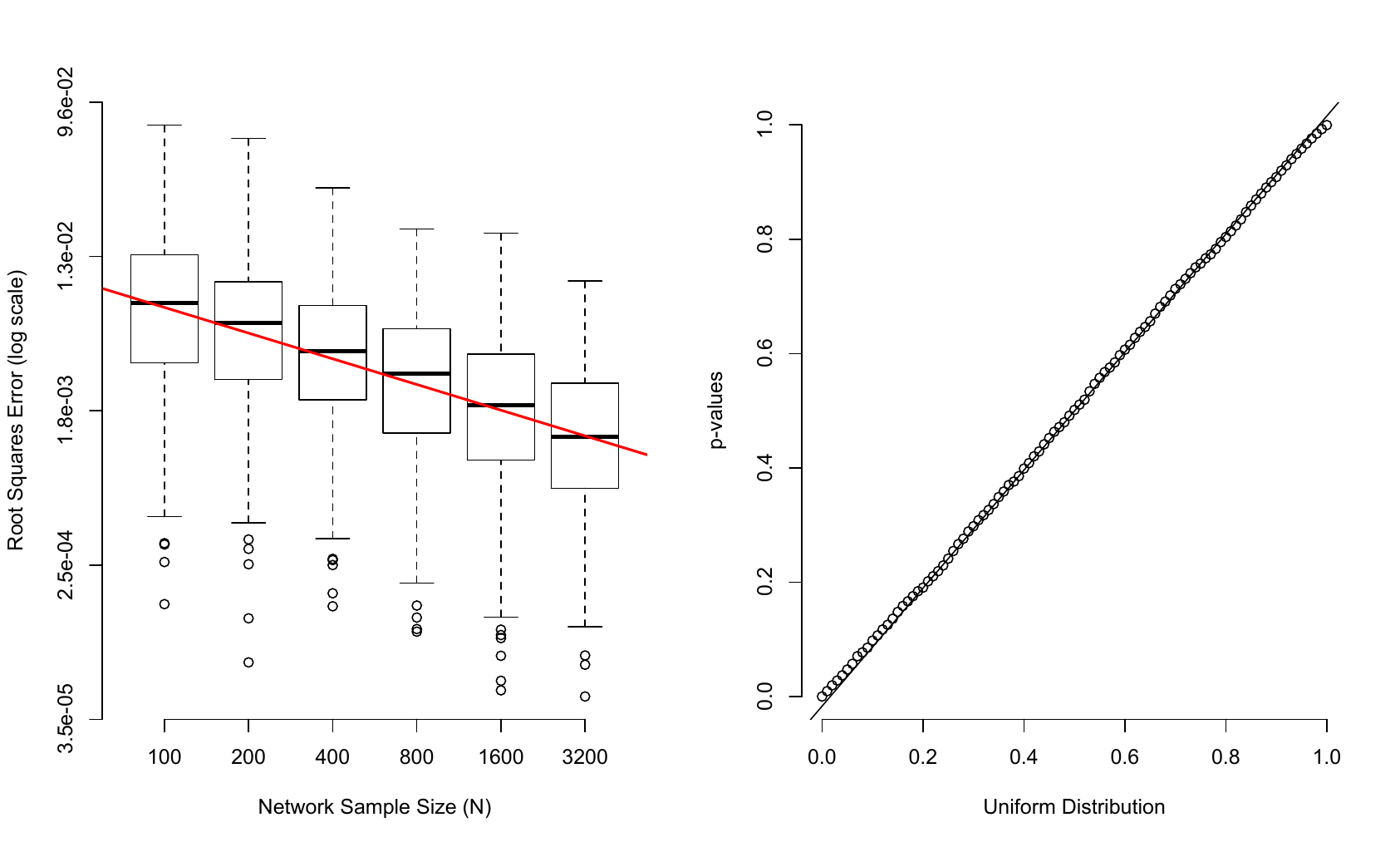}
\caption{Assessing the quality of the asymptotic approximation. We sample $500$ replicates of network samples of sizes $N$ ranging from $100$ to $3200$, each drawn from a random $2$-block blockmodel while the sizes of the networks are fixed by the digits of $\pi$ plus $8$. On each sample, we evaluate the procedure of Figure~\ref{Example1}. We observe that the average mean squared error shrinks as $N$ increases (left plot, trend in red solid line). The observed rate of convergence is in line with our theoretical results. We observe that the QQ-plot of the $p$-values is close to that of the uniform distribution. This QQ-plot aggregates all sample sizes ($N$), but the QQ-plot for $N=100$ is already equivalent to that one.}
\label{here:Experiment}
\end{figure}

\section{Power experiments}
\subsection{Direct application of Theorem~1}
Here we consider the simple problem where we test for a sample being drawn from a given model, as in Figure 2.

Specifically we first fix $k$ and consider the following family of blockmodels $B_k$: graphs over $100$ vertices and $k$ equisized blocks, of fixed global density $\rho = 1/80$, and where between blocks probability of connection is $0$, and within blocks probability of connection is constant (set to that the overall density is $\rho$.)

Then, for $k=2$ to $4$, we draw $12$ graphs from $B_k$, yielding $\G_k$. This enables us to investigate the power of our test. Specifically we estimate the probability to reject when we test for the null of $B_k$ for $\G_{k'}$. To obtain power estimates, we apply Theorem~1 with a level of .05 over 5000 replicates, yielding the following table:

\begin{table}\centering
\renewcommand{\arraystretch}{.8}
\begin{tabular}{cccc}
\toprule
\multirow{2}{*}{Sample name}
&\multicolumn{3}{c}{Null model} \\
\cmidrule(r){2-4}
& \multicolumn{1}{c}{$B_2$}
& \multicolumn{1}{c}{$B_3$}
& \multicolumn{1}{c}{$B_4$} \\
\midrule
$\G_2$ & 0.047 & 0.108 & 0.460\\
$\G_3$ & 0.332 & 0.053 & 0.091\\
$\G_4$ & 0.754 & 0.232 & 0.058\\
\bottomrule
\end{tabular}
\caption{Power when the number of blocks is miss-specified. We find that our test has better power when the null presents fewer blocks than the true model. We also see that as the block get denser with larger $k$, it is computationally harder to maintain the power of the test because of precision with inverting the covariance matrix which is almost singular.}
\end{table}

\subsection{Direct application of Corollary~1}
Resuming from the main document, Section~5. We undertake the following experiment: we estimate kernels (through the random-dot-product framework~\citep{athreya2018survey}) over the samples $\G$, $\G_1^{0.2}$, and $\G_2^{0.2}$, that we call $f$, $f_1$, and $f_2$ respectively; then, for several $\gamma\in[0,1]$, we consider the power of  Corollary~1's test when the samples compared are of the same cardinality as $\G_1^{0.2}$, and $\G_2^{0.2}$, but drawn i.i.d. from $G_{70}\big(\gamma f_1 + (1-\gamma)f\big)$ and $G_{70}\big(\gamma f_2 + (1-\gamma)f\big)$ respectively. 

The purpose of the exercise is to see how large $\gamma$ needs to be for the tests to show reasonably good power. But also to compare the powers of the test to other in the literature. Specifically, we compare to both~\citet{tang2016semipara} and~\citet{tang2016nonpara}. These tests are tightly related to that of~\cite{durante2018}. Call $\beta_{{\tt motif}}$, $ \beta_{{\tt semipar}}$ and $\beta_{{\tt nonpar}}$ the respective power of our test, that of~\citet{tang2016semipara} and that of~\citet{tang2016nonpara} respectively. Then, we observe that:
\begin{itemize}
\item Power is sharply increasing with $\gamma$ for all tests (\cref{plot}).
\item We cannot reject the null that {\tt motif} is at least as powerful as {\tt semipar} for all $\gamma$ (\cref{table}).
\item We can reject the null that {\tt nonpar} is more powerful than {\tt motif} (\cref{table}).
\end{itemize}

\begin{table}[h]
\centering
\renewcommand{\arraystretch}{.8}
\begin{tabular}{ccc}
\toprule
\multirow{2}{*}{$\gamma$}
&\multicolumn{2}{c}{$t$-test outcomes (at $.05$) under the null that} \\
\cmidrule(r){2-3}
& \multicolumn{1}{c}{$\beta_{{\tt motif}} > \beta_{{\tt semipar}}$}
& \multicolumn{1}{c}{$\beta_{{\tt nonpar}} > \beta_{{\tt motif}}$}\\
\midrule
0.1 & Fail to reject & Fail to reject\\
0.2 & Fail to reject & Fail to reject\\
0.3 & Fail to reject & Reject\\
\bottomrule
\end{tabular}
\caption{Outcomes of paired $t$-tests across $\gamma$ under two nulls.}
\label{table}
\end{table}

\begin{figure}[h]
\centering\includegraphics[width=.8\textwidth]{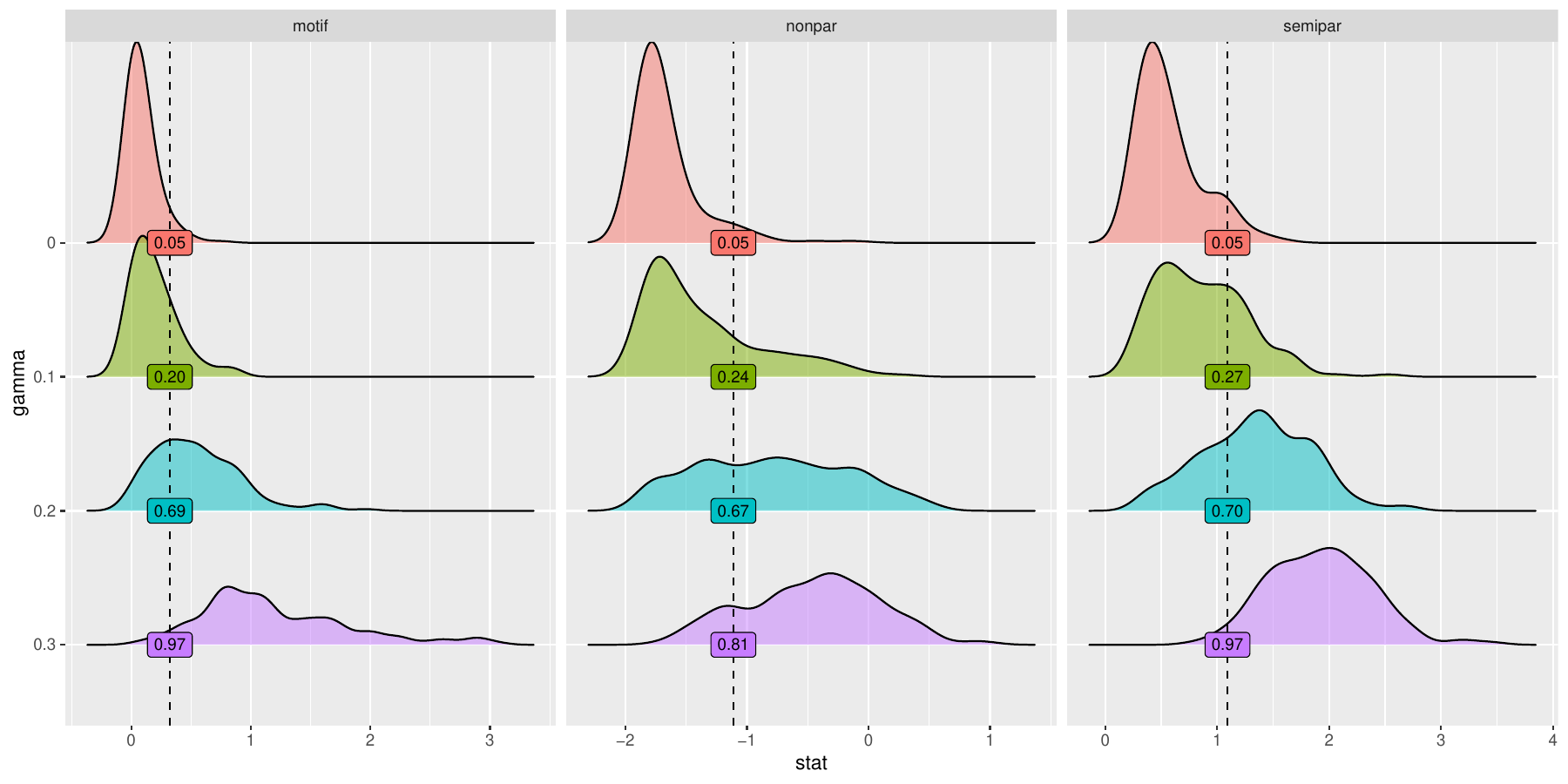}
\caption{Density plots of statistics for each test and $\gamma$. For each $\gamma$ we replicate each test 150 times, each time generating a new sample of graphs.}
\label{plot}
\end{figure}

\section{Methods to produce Figures~\protect\ref{Example1}-\protect\ref{Example3}}
The code to produce all figures and tables is available from the authors. All computation are done in {\tt{R}}. To ease presentation we always denote $K$ the number of blocks in the blockmodel, or flexible blockmodel, under consideration. Furthermore, we call $\pi[k]$ the $k$-th digit of $\pi$.

\subsection{Subgraph counting}
To count the number of copies of \Textedge\ , \Texttriangle\ ,\Textsquare\ , we use the formulas presented in~\cite{alon1997cycles}. Specifically, with $A_{ij}^{(k)}$ the $i,j$-th entry of the $k$-th power of the adjacency matrix of $G$, we use $X_{\Textedge\ \,}(G)   =\ \frac12\sum_{i,j\in[|G|]}A_{ij}$, $X_{\Texttriangle\ }(G) =\ \frac16\sum_{i\in[|G|]}A_{ii}^{(3)}$, and $X_{\Textsquare\ }(G)  =\ \frac18\left(\sum_{i,j\in[|G|]}A^{(3)}_{ij}A_{ij}-4\sum_{i\in[|G|]}\binom{A_{ij}^{(2)}}{2}-2X_{\Textedge\ \,}(G)\right)$. For larger subgraphs we use the results of~\cite{maugis2017fast}.

\subsection{Figure~\protect\ref{Example1}}
We first describe the network sample. The sample is such that for each $i\leq 100$, $G_i\sim G(\pi[i]+30,f_{\bm B})$, where ${\bm B}_{11}=.06$, ${\bm B}_{22}=.06$, ${\bm B}_{12}={\bm B}_{21}=.02$ and $f_{\bm B}(u,v) = {\bm B}_{1_{\{4u\geq 1\}}+1,1_{\{4v\geq 1\}}+1}$.

We now describe the construction of the confidence region and $p$-value under the null. The first task is to compute the subgraph densities (the $\mu_F(f_{\bm B})$ and $\mu_H(f_{\bm B})$). To do so, we use the formulas presented in~\cite{coulson2016poisson,BickelLevina2012} to compute subgraph densities in the blockmodel. More specifically, they show that for a kernel blockmodel $f$ with $K$ blocks such that the probability of being in each block is $\pi_i$, and the block matrix is ${\bm B}$, then
\[
\mu_F(f) = \sum_{i_i,\dots,i_{|F|}\in [K]}\,\prod_{j\in[|F|]}\pi_j\prod_{pq\in F}{\bm B}_{i_pi_q}.
\]
Practically, this is achieved by imbricated loops.

That we can compute the subgraph densities under the null directly implies that we can compute both ${\bm \mu}_\F(f_{\bm B})$ and ${\bm\Sigma}_\F(f_{\bm B})$. Then, the construction of the confidence ellipse as well as of the $p$-value using the Mahalanobis distance are classical statistical inference methods which we need not describe here.

The black region is built with $f_B$. The gray region is built with $f_{{\bm B}'}$ where ${\bm B}_{11}'=.06$, ${\bm B}_{22}'=.05$, ${\bm B}_{12}'={\bm B}_{21}'=.02$ and $f_{{\bm B}'}(u,v) = {\bm B}'_{1_{\{4u\geq 1\}}+1,1_{\{4v\geq 1\}}+1}.$

\subsection{Figure~\protect\ref{Example2}}
We first describe the network sample. The sample is such that for each $i\leq 200$, $G_i\sim G(\pi[k]+30,f_{\bm B})$, where for $i\neq j\in[3]$, ${\bm B}_{ii}=.06$, ${\bm B}_{ij}=.02$ and $f_{\bm B}(u,v) =
{\bm B}_{\lfloor3u\rfloor+1,\lfloor3v\rfloor+1}.$

We now describe how we produced the embedding shape and its confidence region. First, to produce the embedding shape, we compute ${\bm\mu}_\F(f)$ for a a large number of elements of $f\in\D({\bm B})$. Practically, we parametrize $\D({\bm B})$ by a $K$-dimensional vector $\theta$, the entries of which are the sizes of each block, and call the associated kernel $f_\theta$. Then, we build a grid $S$ over the $K$ dimensional simplex of step-size $0.01$, and for each $\theta\in S$ we compute ${\bm\mu}_\F(f_\theta)$. To produce the confidence region, we produce the confidence ellipse for each $\theta\in S$ which we achieve by computing ${\bm\Sigma}_\F(n,f_\theta)$ for each $\theta$ in our grid $S$.

The gray region is built with $f_B$. The black region is built with $f_{{\bm B}'}$ where, for $i\neq j\in[3]$, ${\bm B}_{ii}'=.055$, ${\bm B}_{ij}'=.04$ and $f_{{\bm B}'}(u,v) = {\bm B}'_{\lfloor3u\rfloor+1,\lfloor3v\rfloor+1}.$

\subsection{Figure~\protect\ref{Example3}}
We first describe the network sample. The sample is such that independently for each $i\leq 1000$, $G_i\sim G(\pi[i]+50,f_{\bm B})$, where $f_{\bm B}$ is drawn uniformly at random over $\D({\bm B})$ with for $i\neq j\in[3]$, ${\bm B}_{ii}=.7$, ${\bm B}_{ij}=.2$.

We now describe how we produced the surface where ${\bm\mu}_\F(f)$ may live as well as its confidence region. First, we know that ${\bm\mu}_\F(f)$ may realize any point in the convex hull of the embedding shape. Therefore, we build the embedding shape as for Figure~\ref{Example2}, and then present its convex hull. Building the confidence region is achieved using the following reasoning:
\begin{itemize}
\item[--] Observe that for any kernel $f$ and $H\in\mathcal{H}_{FF'}$, $\mu_H(f)\geq\mu_F(f)\mu_{F'}(f)$. This can for instance be directly recovered from the formulas in Proposition~\ref{first-moms}. Then, for any sequences $n$ and $f$, we have that $({\bm\Sigma}_\F(n,f))_{FF'} \geq 0$, so that all correlations are positive, and the greatest amplitude of the confidence ellipse will be found on its first quadrant; i.e., the variance of ${\bm a}\cdot {\bm\mu}_\F(\G)$ will be maximal for some ${\bm a}\in\mathbb{R}^{|\F|}$ such that ${\bm a}_F>0$ for all $F\in\F$.
\item[--] For any sequences $n$ and $f$, we have that
\begin{align*}
({\bm\Sigma}_\F&(n,f))_{FF'}
=\sum_{F,F'\in\F}\sum_{i\in[N]}\frac{c_HX_H(K_{n_i})}{X_F(K_{n_i})X_{F'}(K_{n_i})}\big(\mu_H(f_i)-\mu_F(f_i)\mu_{F'}(f_i)\big)\\
&\leq\sum_{F,F'\in\F}\sum_{i\in[N]}\frac{c_HX_H(K_{n_i})}{X_F(K_{n_i})X_{F'}(K_{n_i})}\mu_H(f_i)\\
&\leq\sum_{F,F'\in\F}\sum_{i\in[N]}\frac{c_HX_H(K_{n_i})}{X_F(K_{n_i})X_{F'}(K_{n_i})}\max_{f\in\D(B)}\mu_H(f):= ({\bm\Sigma}_\F^\dagger(n,f))_{FF'}.
\end{align*}
Then, for any ${\bm a}$ in the first quadrant we have that ${\bm a}^\top{\bm\Sigma}_\F(n,f) {\bm a}\leq {\bm a}^\top{\bm\Sigma}_\F^\dagger(n,f) {\bm a}$. Therefore, the maximal radius of the ellipse associated with ${\bm\Sigma}_\F^\dagger(n,f)$ on the first quadrant is larger than that induced by ${\bm\Sigma}_\F(n,f)$. Thus, by the previous item, the maximal radius of ${\bm\Sigma}_\F^\dagger(n,f)$ is larger than that of ${\bm\Sigma}_\F(n,f)$. It follows that the auxiliary sphere of the ellipse associated with ${\bm\Sigma}_\F^\dagger(n,f)$ contains that associated with ${\bm\Sigma}_\F(n,f)$.
\item[--] Finally, $\max_{f\in\D({\bm B})}\mu_H(f)$ is dominated by $\mu_H(f_\ast)$, where $f_\ast$ is the constant kernel equal to the maximal entry of $B$. In this final setting, we need only compute the ${\bm\mu}_H(f_\ast)$, where $f_\ast$ describes an Erd\H{o}s-R\'enyi random graph. This makes the computation straightforward as we then have ${\bm\mu}_H(f_\ast) = \max_{ij}{\bm B}_{ij}^k$, where $k$ is the number of edges in $H$.
\end{itemize}
Following this argument, our confidence region is the union of the auxiliary spheres of the ellipse associated with ${\bm\Sigma}_\F^\dagger(n,f)$ centered at each point in the convex hull of the embedding shape. Finally, the gray region is built with ${\bm B}'$, where for $i\neq j\in[4]$, ${\bm B}_{ii}'=.8$, ${\bm B}_{ij}=.55$.
\bibliographystyle{chicago}
\bibliography{NetworkSample_bib}

\begin{thebibliography}{}

\bibitem[\protect\citeauthoryear{Ali, Rito, Reinert, Sun, and Deane}{Ali
  et~al.}{2014}]{Ali2014alignment}
Ali, T., T.~Rito, G.~Reinert, F.~Sun, and C.~M. Deane (2014).
\newblock Alignment-free protein interaction network comparison.
\newblock {\em Bioinformatics\/}~{\em 30}, i430--i437.

\bibitem[\protect\citeauthoryear{Ali, Wegner, Gaunt, Deane, and Reinert}{Ali
  et~al.}{2016}]{ali2016comparison}
Ali, W., A.~E. Wegner, R.~E. Gaunt, C.~M. Deane, and G.~Reinert (2016).
\newblock Comparison of large networks with sub--sampling strategies.
\newblock {\em Scientific Reports\/}~{\em 6}, 28955.

\bibitem[\protect\citeauthoryear{Alon, Yuster, and Zwick}{Alon
  et~al.}{1997}]{alon1997cycles}
Alon, N., R.~Yuster, and U.~Zwick (1997).
\newblock Finding and counting given length cycles.
\newblock {\em Algorithmica\/}~{\em 17\/}(3), 209--223.

\bibitem[\protect\citeauthoryear{Asta and Shalizi}{Asta and
  Shalizi}{2014}]{asta2014geometric}
Asta, D. and C.~R. Shalizi (2014).
\newblock Geometric network comparison.
\newblock arXiv:1411.1350.

\bibitem[\protect\citeauthoryear{Athreya, Fishkind, Levin, Lyzinski, Park, Qin,
  Sussman, Tang, Vogelstein, and Priebe}{Athreya
  et~al.}{2018}]{athreya2018survey}
Athreya, A., D.~E. Fishkind, K.~Levin, V.~Lyzinski, Y.~Park, Y.~Qin, D.~L.
  Sussman, M.~Tang, J.~T. Vogelstein, and C.~E. Priebe (2018).
\newblock Statistical inference on random dot product graphs: a survey.
\newblock {\em JMLR\/}~{\em 18}, 1--92.

\bibitem[\protect\citeauthoryear{Banerjee and Ma}{Banerjee and
  Ma}{2017}]{banerjee2017testing}
Banerjee, D. and Z.~Ma (2017).
\newblock Optimal hypothesis testing for stochastic block models with growing
  degrees.
\newblock arXiv:1705.05305.

\bibitem[\protect\citeauthoryear{Benson, Gleich, and Leskovec}{Benson
  et~al.}{2016}]{benson2016higher}
Benson, A.~R., D.~F. Gleich, and J.~Leskovec (2016).
\newblock Higher-order organization of complex networks.
\newblock {\em Science\/}~{\em 353\/}(6295), 163--166.

\bibitem[\protect\citeauthoryear{Bhattacharyya and Bickel}{Bhattacharyya and
  Bickel}{2015}]{bhattacharyya2013subsampling}
Bhattacharyya, S. and P.~J. Bickel (2015).
\newblock Subsampling bootstrap of count features of networks.
\newblock {\em Ann. Statist.\/}~{\em 43}, 2384--2411.

\bibitem[\protect\citeauthoryear{Bickel, Chen, and Levina}{Bickel
  et~al.}{2012}]{BickelLevina2012}
Bickel, P.~J., A.~Chen, and E.~Levina (2012).
\newblock The method of moments and degree distributions for network models.
\newblock {\em Ann. Statist.\/}~{\em 39}, 2280--2301.

\bibitem[\protect\citeauthoryear{Birmele}{Birmele}{2012}]{birmele2012detecting}
Birmele, E. (2012).
\newblock Detecting local network motifs.
\newblock {\em E. J. Stat.\/}~{\em 6}, 908--933.

\bibitem[\protect\citeauthoryear{Bollob{\'a}s and Riordan}{Bollob{\'a}s and
  Riordan}{2009}]{bollobas2009metric}
Bollob{\'a}s, B. and O.~Riordan (2009).
\newblock Metrics for sparse graphs.
\newblock In S.~Huczynska, J.~D. Mitchell, and C.~M. Roney-Dougal (Eds.), {\em
  Surveys in Combinatorics 2009}, pp.\  211--287. Cambridge, UK: Cambridge
  University Press.

\bibitem[\protect\citeauthoryear{Chatterjee and Diaconis}{Chatterjee and
  Diaconis}{2013}]{chatterjee2013estimating}
Chatterjee, S. and P.~Diaconis (2013).
\newblock Estimating and understanding exponential random graph models.
\newblock {\em Ann. Statist.\/}~{\em 41\/}(5), 2428--2461.

\bibitem[\protect\citeauthoryear{Coulson, Gaunt, and Reinert}{Coulson
  et~al.}{2016}]{coulson2016poisson}
Coulson, M., R.~E. Gaunt, and G.~Reinert (2016).
\newblock Poisson approximation of subgraph counts in stochastic block models
  and a graphon model.
\newblock arXiv:1509.07754.

\bibitem[\protect\citeauthoryear{Daianu, Jahanshad, Nir, Toga, Jack, Weiner,
  and Thompson}{Daianu et~al.}{2013}]{daianu2013}
Daianu, M., N.~Jahanshad, T.~Nir, A.~Toga, C.~Jack, M.~Weiner, and P.~Thompson
  (2013).
\newblock Breakdown of brain connectivity between normal aging and alzheimer's
  disease: A structural k-core network analysis.
\newblock {\em Brain Connectivity\/}~{\em 3}, 407--422.

\bibitem[\protect\citeauthoryear{Desikan, Segonne, Fischl, Quinn, Dickerson,
  Blacker, Buckner, Dale, Maguire, Hyman, Albert, and Killiany}{Desikan
  et~al.}{2006}]{desikan}
Desikan, R., F.~Segonne, B.~Fischl, B.~Quinn, B.~Dickerson, D.~Blacker,
  R.~Buckner, A.~Dale, R.~Maguire, B.~Hyman, M.~Albert, and R.~Killiany (2006).
\newblock An automated labeling system for subdividing the human cerebral
  cortex on mri scans into gyral based regions of interest.
\newblock {\em NeuroImage\/}~{\em 16530430}, 1053--8119.

\bibitem[\protect\citeauthoryear{Diaconis and Janson}{Diaconis and
  Janson}{2008}]{diaconis2008graph}
Diaconis, P. and S.~Janson (2008).
\newblock Graph limits and exchangeable random graphs.
\newblock {\em Rendi. Mat. Appl.\/}~{\em 28}, 33--61.

\bibitem[\protect\citeauthoryear{Durante and Dunson}{Durante and
  Dunson}{2018}]{durante2018}
Durante, D. and D.~B. Dunson (2018).
\newblock Bayesian inference and testing of group differences in brain
  networks.
\newblock {\em Bayesian Anal.\/}~{\em 13}, 29--58.

\bibitem[\protect\citeauthoryear{Fosdick and Hoff}{Fosdick and
  Hoff}{2015}]{fosdick2015testing}
Fosdick, B. and P.~D. Hoff (2015).
\newblock Testing and modeling dependencies between a network and nodal
  attributes.
\newblock {\em J. Amer. Statist. Assoc.\/}~{\em 110\/}(511), 1047--1056.

\bibitem[\protect\citeauthoryear{Gao and Lafferty}{Gao and
  Lafferty}{2017}]{gao2017testing}
Gao, C. and J.~Lafferty (2017).
\newblock Testing network structure using relations between small subgraph
  probabilities.
\newblock arXiv:1704.06742.

\bibitem[\protect\citeauthoryear{Garyfallidis, Brett, Amirbekian, Rokem, Van
  Der~Walt, Descoteaux, and Nimmo-Smith}{Garyfallidis et~al.}{2014}]{dipy}
Garyfallidis, E., M.~Brett, B.~Amirbekian, A.~Rokem, S.~Van Der~Walt,
  M.~Descoteaux, and I.~Nimmo-Smith (2014).
\newblock Dipy, a library for the analysis of diffusion mri data.
\newblock {\em Frontiers in neuroinformatics\/}~{\em 8}, 8.

\bibitem[\protect\citeauthoryear{Garyfallidis, Brett, Correia, Williams, and
  Nimmo-Smith}{Garyfallidis et~al.}{2012}]{eudx}
Garyfallidis, E., M.~Brett, M.~Correia, G.~Williams, and I.~Nimmo-Smith (2012).
\newblock Quickbundles, a method for tractography simplification.
\newblock {\em Frontiers in neuroscience\/}~{\em 6}, 175.

\bibitem[\protect\citeauthoryear{Ghoshdastidar, von Luxburg, Gutzeit, and
  Carpentier}{Ghoshdastidar et~al.}{2017}]{ghoshdastidar2017testing}
Ghoshdastidar, D., U.~von Luxburg, M.~Gutzeit, and A.~Carpentier (2017).
\newblock Two-sample tests for large random graphs using network statistics.
\newblock arXiv:1705.06168.

\bibitem[\protect\citeauthoryear{Ginestet, Li, Balachandran, Rosenberg, and
  Kolaczyk}{Ginestet et~al.}{2017}]{ginestet2017}
Ginestet, C.~E., J.~Li, P.~Balachandran, S.~Rosenberg, and E.~D. Kolaczyk
  (2017).
\newblock Hypothesis testing for network data in functional neuroimaging.
\newblock {\em Ann. Appl. Stat.\/}~{\em 11}, 725--750.

\bibitem[\protect\citeauthoryear{Gray, Bogovic, Vogelstein, Landman, Prince,
  and Vogelstein}{Gray et~al.}{2012}]{gray2012magnetic}
Gray, W.~R., J.~A. Bogovic, J.~T. Vogelstein, B.~A. Landman, J.~L. Prince, and
  R.~J. Vogelstein (2012).
\newblock Magnetic resonance connectome automated pipeline: an overview.
\newblock {\em IEEE pulse\/}~{\em 3\/}(2), 42--48.

\bibitem[\protect\citeauthoryear{Ho, Parikh, and Xing}{Ho
  et~al.}{2012}]{ho2012multiscale}
Ho, Q., A.~P. Parikh, and E.~P. Xing (2012).
\newblock A multiscale community blockmodel for network exploration.
\newblock {\em J. Amer. Statist. Assoc.\/}~{\em 107\/}(499), 916--934.

\bibitem[\protect\citeauthoryear{Hoff, Raftery, and Handcock}{Hoff
  et~al.}{2012}]{hoff2012latent}
Hoff, P.~D., A.~E. Raftery, and M.~S. Handcock (2012).
\newblock Latent space approaches to social network analysis.
\newblock {\em J. Amer. Statist. Assoc.\/}~{\em 97\/}(460), 1090--1098.

\bibitem[\protect\citeauthoryear{Ho\v{c}evar and Dem\v{s}ar}{Ho\v{c}evar and
  Dem\v{s}ar}{2014}]{hocevar2014}
Ho\v{c}evar, T. and J.~Dem\v{s}ar (2014).
\newblock A combinatorial approach to graphlet counting.
\newblock {\em Bioinformatics\/}~{\em 30}, 559--565.

\bibitem[\protect\citeauthoryear{Jenkinson, Beckmann, Behrens, M., and
  S.}{Jenkinson et~al.}{2012}]{fsl3}
Jenkinson, M., C.~Beckmann, T.~Behrens, W.~M., and S.~S. (2012).
\newblock {FSL.}
\newblock {\em NeuroImage\/}~{\em 62\/}(2), 782--90.

\bibitem[\protect\citeauthoryear{Jha, Seshadhri, and Pinar}{Jha
  et~al.}{2015}]{jha2015path}
Jha, M., C.~Seshadhri, and A.~Pinar (2015).
\newblock Path sampling: A fast and provable method for estimating 4-vertex
  subgraph counts.
\newblock In {\em Proceedings of the 24th International Conference on World
  Wide Web}, WWW '15, New York, NY, USA, pp.\  495--505. ACM.

\bibitem[\protect\citeauthoryear{Kiar, Bridgeford, Roncal, (CoRR),
  Reproducibliity, Chandrashekhar, Mhembere, Ryman, Zuo, Marguiles, Craddock,
  Priebe, Jung, Calhoun, Caffo, Burns, Milham, and Vogelstein}{Kiar
  et~al.}{2018}]{ndmg}
Kiar, G., E.~Bridgeford, W.~Roncal, C.~f.~R. (CoRR), Reproducibliity,
  V.~Chandrashekhar, D.~Mhembere, S.~Ryman, X.~Zuo, D.~Marguiles, R.~Craddock,
  C.~Priebe, R.~Jung, V.~Calhoun, B.~Caffo, R.~Burns, M.~Milham, and
  J.~Vogelstein (2018).
\newblock A high-throughput pipeline identifies robust connectomes but
  troublesome variability.
\newblock {\em bioRxiv\/}, 188706.

\bibitem[\protect\citeauthoryear{Klopp, Tsybakov, and Verzelen}{Klopp
  et~al.}{2016}]{klopp2016}
Klopp, O., A.~B. Tsybakov, and N.~Verzelen (2016).
\newblock Oracle inequalities for network models and sparse graphon estimation.
\newblock {\em Ann. Statist.\/}~{\em 45}, 316--354.

\bibitem[\protect\citeauthoryear{Koutra, Vogelstein, and Faloutsos}{Koutra
  et~al.}{2013}]{koutra2013deltacon}
Koutra, D., J.~T. Vogelstein, and C.~Faloutsos (2013).
\newblock {\protect\sc{DeltaCon}}: A principled massive-graph similarity
  function.
\newblock In {\em Proceedings of the SIAM International Conference in Data
  Mining. Society for Industrial and Applied Mathematics}, pp.\  162--170.
  SIAM.

\bibitem[\protect\citeauthoryear{Lov{\'a}sz}{Lov{\'a}sz}{2012}]{lovasz2012large}
Lov{\'a}sz, L. (2012).
\newblock {\em Large Networks and Graph Limits}.
\newblock Providence, RI: Am. Mathematical Soc.

\bibitem[\protect\citeauthoryear{Maugis, Olhede, and Wolfe}{Maugis
  et~al.}{2017a}]{maugis2017fast}
Maugis, P.-A.~G., S.~C. Olhede, and P.~J. Wolfe (2017a).
\newblock Fast counting of medium-sized rooted subgraphs.
\newblock arXiv:1701.00177.

\bibitem[\protect\citeauthoryear{Maugis, Olhede, and Wolfe}{Maugis
  et~al.}{2017b}]{maugis2017topo}
Maugis, P.-A.~G., S.~C. Olhede, and P.~J. Wolfe (2017b).
\newblock Topology reveals universal features for network comparison.
\newblock arXiv:1705.056777.

\bibitem[\protect\citeauthoryear{Mazziotta, Toga, Evans, Fox, Lancaster,
  Zilles, Woods, Paus, Simpson, Pike, Holmes, Collins, Thompson, MacDonald,
  Iacoboni, Schormann, Amunts, Palomero-Gallagher, Geyer, Parsons, Narr,
  Kabani, Le~Goualher, Feidler, Smith, Boomsma, Hulshoff, Cannon, Kawashima,
  and Mazoyer}{Mazziotta et~al.}{2001}]{mni152}
Mazziotta, J., A.~Toga, A.~Evans, P.~Fox, J.~Lancaster, K.~Zilles, R.~Woods,
  T.~Paus, G.~Simpson, B.~Pike, C.~Holmes, L.~Collins, P.~Thompson,
  D.~MacDonald, M.~Iacoboni, T.~Schormann, K.~Amunts, N.~Palomero-Gallagher,
  S.~Geyer, L.~Parsons, K.~Narr, N.~Kabani, G.~Le~Goualher, J.~Feidler,
  K.~Smith, D.~Boomsma, P.~Hulshoff, T.~Cannon, R.~Kawashima, and B.~Mazoyer
  (2001).
\newblock A four-dimensional probabilistic atlas of the human brain.
\newblock {\em Journal of the American Medical Informatics Association\/}~{\em
  8}, 401--430.

\bibitem[\protect\citeauthoryear{Milo, Shen-Orr, Itzkovitz, Kashtan,
  Chklovskii, and Alon}{Milo et~al.}{2002}]{milo2002network}
Milo, R., S.~Shen-Orr, S.~Itzkovitz, N.~Kashtan, D.~Chklovskii, and U.~Alon
  (2002).
\newblock Network motifs: Simple building blocks of complex networks.
\newblock {\em Science\/}~{\em 298}, 824--827.

\bibitem[\protect\citeauthoryear{Olhede and Wolfe}{Olhede and
  Wolfe}{2014}]{olhede2013network}
Olhede, S.~C. and P.~J. Wolfe (2014).
\newblock Network histograms and universality of blockmodel approximation.
\newblock {\em Proc. Natl. Acad. Sci. USA\/}~{\em 111}, 14722--14727.

\bibitem[\protect\citeauthoryear{Ortmann and Brandes}{Ortmann and
  Brandes}{2016}]{ortmann2016quad}
Ortmann, M. and U.~Brandes (2016).
\newblock {\em Quad census computation: simple, efficient, and orbit-aware},
  pp.\  1--13.
\newblock Cham: Springer Int. Publishing.

\bibitem[\protect\citeauthoryear{Pinar, Seshadhri, and Vishal}{Pinar
  et~al.}{2016}]{pinar2016escape}
Pinar, A., C.~Seshadhri, and V.~Vishal (2016).
\newblock Escape: Efficiently counting all 5-vertex subgraphs.
\newblock arXiv:1610.09411.

\bibitem[\protect\citeauthoryear{Rucinski}{Rucinski}{1988}]{rucinski1988small}
Rucinski, A. (1988).
\newblock When are small subgraphs of a random graph normally distributed?
\newblock {\em Probab. Theory Related Fields\/}~{\em 78\/}(1), 1--10.

\bibitem[\protect\citeauthoryear{Simpson, Moussa, and Laurienti}{Simpson
  et~al.}{2012}]{simpson2012exponential}
Simpson, S., M.~Moussa, and P.~Laurienti (2012).
\newblock An exponential random graph modeling approach to creating group-based
  representative whole-brain connectivity networks.
\newblock {\em NeuroImage\/}~{\em 60}, 1117--1126.

\bibitem[\protect\citeauthoryear{Smith, Jenkinson, Woolrich, Beckmann, Behrens,
  Johansen-Berg, Bannister, De~Luca, Drobnjak, Flitney, Niazy, Saunders,
  Vickers, Zhang, De~Stefano, Brady, and Matthews}{Smith et~al.}{2004}]{fsl1}
Smith, S., M.~Jenkinson, M.~Woolrich, C.~Beckmann, T.~Behrens,
  H.~Johansen-Berg, P.~Bannister, M.~De~Luca, I.~Drobnjak, D.~Flitney,
  R.~Niazy, J.~Saunders, J.~Vickers, Y.~Zhang, N.~De~Stefano, J.~Brady, and
  P.~Matthews (2004).
\newblock Advances in functional and structural mr image analysis and
  implementation as fsl.
\newblock {\em NeuroImage\/}~{\em 23}, S208--19.

\bibitem[\protect\citeauthoryear{Stoffers, Berendse, Olde~Dubbelink,
  Hillebrand, Stam, Deijen, and Twisk}{Stoffers et~al.}{2013}]{stoffers2013}
Stoffers, D., H.~Berendse, K.~Olde~Dubbelink, A.~Hillebrand, C.~Stam,
  J.~Deijen, and J.~Twisk (2013).
\newblock Disrupted brain network topology in parkinson’s disease: a
  longitudinal magnetoencephalography study.
\newblock {\em Brain\/}~{\em 137}, 197--207.

\bibitem[\protect\citeauthoryear{Sussman, Tang, Fishkind, and Priebe}{Sussman
  et~al.}{2012}]{Sussman2012}
Sussman, D.~L., M.~Tang, D.~E. Fishkind, and C.~E. Priebe (2012).
\newblock A consistent adjacency spectral embedding for stochastic blockmodel
  graphs.
\newblock {\em J. Amer. Statist. Assoc.\/}~{\em 107}, 1119--1128.

\bibitem[\protect\citeauthoryear{Talukder and Zaki}{Talukder and
  Zaki}{2016}]{talukder2016distributed}
Talukder, N. and M.~J. Zaki (2016).
\newblock A distributed approach for graph mining in massive networks.
\newblock {\em Data Min. Knowl. Discov.\/}~{\em 30}, 1024--1052.

\bibitem[\protect\citeauthoryear{Tang, Athreya, Sussman, Lyzinski, Park, and
  Priebe}{Tang et~al.}{2017}]{tang2016semipara}
Tang, M., A.~Athreya, D.~L. Sussman, V.~Lyzinski, Y.~Park, and C.~E. Priebe
  (2017).
\newblock A semiparametric two-sample hypothesis testing for random dot product
  graphs.
\newblock {\em Journal of Computational and Graphical Statistics\/}~{\em 26},
  344--354.

\bibitem[\protect\citeauthoryear{Tang, Athreya, Sussman, Lyzinski, and
  Priebe}{Tang et~al.}{2017}]{tang2016nonpara}
Tang, M., A.~Athreya, D.~L. Sussman, V.~Lyzinski, and C.~E. Priebe (2017).
\newblock A nonparametric two-sample hypothesis testing for random dot product
  graphs.
\newblock {\em Bernoulli\/}~{\em 23}, 1599--1630.

\bibitem[\protect\citeauthoryear{Tang, Sussman, and Priebe}{Tang
  et~al.}{2013}]{tang2013}
Tang, M., D.~L. Sussman, and C.~E. Priebe (2013).
\newblock Universally consistent vertex classification for latent positions
  graphs.
\newblock {\em Ann. Statist.\/}~{\em 41}, 1406--1430.

\bibitem[\protect\citeauthoryear{Wang, Zhang, and Dunson}{Wang
  et~al.}{2019}]{wang2018}
Wang, L., Z.~Zhang, and D.~Dunson (2019).
\newblock Common and individual structure of brain networks.
\newblock {\em Ann. App. Stat.\/}~{\em 13}, 85--112.

\bibitem[\protect\citeauthoryear{Woolrich, Jbabdi, Patenaude, Chappell, Makni,
  Behrens, Beckmann, Jenkinson, and Smith}{Woolrich et~al.}{2009}]{fsl2}
Woolrich, M., S.~Jbabdi, B.~Patenaude, M.~Chappell, S.~Makni, T.~Behrens,
  C.~Beckmann, M.~Jenkinson, and S.~Smith (2009).
\newblock Bayesian analysis of neuroimaging data in fsl.
\newblock {\em NeuroImage\/}~{\em 45}, S173--86.

\bibitem[\protect\citeauthoryear{Zuo, Anderson, Bellec, Birn, Biswal, Blautzik,
  Breitner, Buckner, Calhoun, Castellanos, et~al.}{Zuo et~al.}{2014}]{corr}
Zuo, X.-N., J.~S. Anderson, P.~Bellec, R.~M. Birn, B.~B. Biswal, J.~Blautzik,
  J.~C. Breitner, R.~L. Buckner, V.~D. Calhoun, F.~X. Castellanos, et~al.
  (2014).
\newblock An open science resource for establishing reliability and
  reproducibility in functional connectomics.
\newblock {\em Scientific data\/}~{\em 1}, 140049.

\end{thebibliography}
\end{document}